\documentclass[aip,amsmath,amssymb,reprint]{revtex4-1}
\usepackage{graphicx}
\usepackage{dcolumn}
\usepackage{bm}
\usepackage[mathlines]{lineno}

\usepackage[utf8]{inputenc}
\usepackage[T1]{fontenc}
\usepackage{mathptmx}
\usepackage{xcolor}
\usepackage{graphicx}
\usepackage{float}
\usepackage{braket} 
\usepackage{mathtools}
\usepackage{graphicx,wrapfig,lipsum}
\usepackage{array}
\newcolumntype{P}[1]{>{\centering\arraybackslash}p{#1}}
\newcolumntype{M}[1]{>{\centering\arraybackslash}m{#1}}
\usepackage{hyphsubst}
\usepackage{amsmath,amstext,amsthm,amscd}
\usepackage{appendix}

\theoremstyle{plain}
\newtheorem{theorem}{Theorem}[section]
\newtheorem{lemma}[theorem]{Lemma}
\newtheorem{corollary}[theorem]{Corollary}
\newtheorem{proposition}[theorem]{Proposition}
\newtheorem{question}[theorem]{Question}
\newtheorem{definition}[theorem]{Definition}

\newtheorem{example}[theorem]{Example}

\newtheorem{remark}[theorem]{Remark}
\numberwithin{equation}{section}

\newcommand{\R}{\mathbb{R}}
\newcommand{\N}{\mathbb N}
\usepackage{xcolor}

\begin{document}

\title{Generic Features in the Spectral Decomposition of Correlation Matrices}

\author{Yuriy Stepanov}
\email{yuriy.stepanov@uni-due.de}
\affiliation{Faculty of Physics, University of Duisburg-Essen,  Duisburg, Germany.}

\author{Hendrik Herrmann}
\email{hherrmann@uni-wuppertal.de}
\affiliation{Department of Mathematics, Wuppertal University, Wuppertal, Germany.}

\author{Thomas Guhr}
\email{thomas.guhr@uni-due.de}
\affiliation{Faculty of Physics, University of Duisburg-Essen,  Duisburg, Germany.}

\date{\today}

\begin{abstract}
We show that correlation matrices with particular average and  variance of the correlation coefficients have a notably restricted  spectral structure. Applying geometric methods, we derive  lower bounds for the largest eigenvalue and the alignment of the corresponding eigenvector. We explain how and to which extent, a distinctly large eigenvalue and an approximately diagonal eigenvector generically occur for specific correlation matrices independently of the correlation matrix dimension.

\end{abstract}

\maketitle

\tableofcontents

\section{Introduction} \label{sec:intro}

To catch up  with  ever increasing complexity of technologies and the nowadays available data amounts,  empirical scientists chase for patterns of collective behaviour. Reducing complexity of the models has a large value, especially in the industry. In the widely used multivariate methods, the covariance and  correlation matrices  play the central role \cite{Tinsley2000}. Correlation matrices   are  used across different sciences \cite{archdeacon1994correlation,EggCorrelations2014,Batushansky2016}  and especially in finance  \cite{Markowitz1956,POLLET2010,meissner2013correlation,ROUKNY2018}. 

On the mathematical side,  stochastic approaches, like random matrix theory\cite{Fyodorov2011,BunBouchaudPotters2017}, have strongly influenced these fields in the past hundred years, mostly focusing on the spectral structure \cite{Wishart1928,Wigner1967,MarcenkoPastur1967,Pastur1973,Edwards1976,Friedman1981,Furedi_Komlos_1981,Juhasz1981,Homes_1991} of random matrices. Usually, the results are stated as statistical limits or only apply for infinitely large correlation matrices. A practitioner  is often unable  to quantify, to which extent a model is applicable to particular empirical data.

Recent empirical studies of financial correlation have clearly shown its non-stationarity \cite{CONLON2009,Song2011,Munnix2012}.  At the same time, financial correlation matrices have been repeatedly reported\cite{Laloux1999,Plerou1999,Plerou2002} to generically have  an approximately diagonal eigenvector, corresponding to a distinctly large eigenvalue. Furthermore, for empirical ensembles of correlation matrices, the largest eigenvalue is proportional to the average correlation\cite{Song2011,Stepanov_2015}. Similar observations have been made for simulated data\cite{Friedman1981,SornetteMalevergne2004} as well. Such correlation matrices are approximately determined by a single eigenvector, which is a notable simplification.

Correlation matrices with constant non-zero coefficients, trivially have these generic features, whatever the matrix dimension is \cite{Morrison1976, Friedman1981}. As shown by Füredi \textit{et~al.}\cite{Furedi_Komlos_1981} and Malevergne \textit{et~al.}\cite{SornetteMalevergne2004}, the same applies to infinitely large correlation matrices in statistical limit provided the average correlation is positive, and the variance of the correlation coefficients is small enough.

The scope of the present study is to understand under which conditions and to which extent, arbitrary correlation matrices have these generic features. Hence we focus on the relation between the average correlation and the variance of the correlation coefficients to the spectral structure of the underlying correlation matrix. Furthermore, we address the impact of the correlation matrix dimension.

In the present paper we extend the results of Refs.\cite{Furedi_Komlos_1981,SornetteMalevergne2004} to  correlation matrices of an arbitrary dimension.   We show that the average correlation and the variance of the correlation coefficients imply constraints on the spectrum of the underlying correlation matrix. Applying methods from linear algebra,  we derive lower bounds on the largest eigenvalue and restrictions on the alignment of the corresponding eigenvector. We show that  no matter matrix dimension is, a distinctly large eigenvalue with an approximately diagonal vector, simultaneously occur for a wide range of correlation matrices.

The paper is organised as follows:  In Sec.~\ref{sec:char_lemma} we discus general properties of correlation matrices and explain our methods. We state our main results in Sec.~\ref{sec:main_reuslts}. We give the proofs of the main results in the rather technical sections Secs. \ref{sec:proofs_Th1}--\ref{sec:ProofTheorem4}. In Sec.~\ref{sec:conclusion} we conclude our findings.

\section{Characteristic Lemma and Examples of Correlation Matrices} \label{sec:char_lemma}

In  Sec.~\ref{sec:genertal_corr_mat} we summarise general properties of correlation matrices. We consider distinct examples in Sec.~\ref{sec:genertal_corr_mat}. We derive the characteristic lemma of correlation matrices and explain our methods  in Sec.~\ref{sec:methods}.

\subsection{General Features of Correlation Matrices} \label{sec:genertal_corr_mat}

We introduce correlation matrices from the geometric point of view. For $n \geq 2$ and two vectors \(x,y\in\R^n\), their standard inner product is defined by
	
\begin{equation}
\langle x,y\rangle:=\sum_{j=1}^n x_jy_j.
\end{equation}
We denote the Euclidean norm of \(x\) by \(\|x\|:=\sqrt{\langle x,x\rangle}\). Consider  an $N \times n$ matrix
\begin{equation}
		M:=[r_1,\ldots,r_n].
\end{equation} 
The columns  $r_1,\ldots, r_n \in \R^N$ of $M$ are  $n$ arbitrary vectors, normalised by
\begin{equation}
\parallel 	 r_i \parallel  \equiv 1,
\end{equation} for $1 \leq i\leq n$.
Given a matrix $M$, we define the  $n \times n$   matrix 
\begin{equation} \label{eq:def_C}
		C:=M^TM,
\end{equation} as the  product of matrix $M$ and its transpose $M^T$. Its coefficients  
\begin{equation} 
C_{ij}=\langle r_i,r_j \rangle,
\end{equation} are the pairwise inner products of the vectors $r_1,\ldots,
 r_n$.
Matrix $C$, as defined in \eqref{eq:def_C}, has three characteristic properties:

\begin{itemize}
	\item[(i)] $C$ is symmetric,\textit{ i.e.}  \(C_{ij}=C_{ji}\)  for \(1\leq i,j\leq n\),
	\item[(ii)] \(C_{ii}=1\) for \(1\leq i\leq n\),
	\item[(iii)] $C$ is positive semi-definite, \textit{i.e.} $x^TCx\geq 0$ for all $x\in \mathbb R^n$.
\end{itemize}	
In most cases,  the matrix $M$ and hence $C$ are random variables\cite{Tinsley2000}. In the present study, we refer to a correlation matrix  $C$ as a real $n \times n$ matrix, which fulfils conditions (i)-(iii). In particular we consider a correlation matrix $C$ as a fixed realisation of a random variable and we don't refer to any matrix $M$. We note that any correlation matrix $C$ can be written\cite{Stefanica2014} as in \eqref{eq:def_C} and one automatically has\cite{Stefanica2014} 
\begin{equation} \label{eq:abs_cij_lower_1}
 C_{ij}\in[-1,1],
\end{equation}
for $1\leq i,j\leq n$. Conditions (i)-(iii) imply further  properties of correlation matrices. From (i) one has that any $n \times n$ correlation matrix $C$ can be spectrally decomposed
\begin{equation} \label{eq:C_spectral_decomposition}
		C=\sum_{i=1}^{n}\lambda_i v_i v_i^T.
\end{equation}
Here $\lambda_1\geq\ldots\geq\lambda_n$ are the real eigenvalues and $v_1,..,v_n \in \R^n$ is an orthonormal basis consisting of corresponding eigenvectors, i.e.~$  Cv_i=\lambda_iv_i$, \(1\leq i\leq n\). We introduce the normalised diagonal vector
\begin{equation}
\delta_n:=\frac{1}{\sqrt{n}}(1,\ldots,1) \in \R^n,
\end{equation} which is distinctly important in the following.
 We note that if one of the eigenvectors is parallel to $\delta_n$, the remaining eigenvectors are orthogonal to it due to orthogonality of the eigenbasis.

From (ii) one has for the trace $\text{Tr}(C) =\sum_{i=1}^{n}C_{ii} = n$. Hence for the eigenvalues one has
\begin{equation} \label{eq:trace}
\sum_{i=1}^{n}\lambda_i = n.
\end{equation} Therefore, one has  with condition (iii) for the eigenvalue range
\begin{equation}
  0 \leq\lambda_i\leq n,
\end{equation} for $1 \leq i\leq n$. 

From \eqref{eq:C_spectral_decomposition} we observe that correlation matrices with a strongly dominating eigenvalue $\lambda_1 \gg \lambda_2$ are approximately given by the single eigenvector $v_1$.
 Especially for empirical applications, it is crucial to understand which correlation matrices have a distinctly large eigenvalue and what is the corresponding eigenvector.

Motivated by previous studies\cite{Furedi_Komlos_1981,SornetteMalevergne2004}, we characterise an $n \times n$ correlation matrix $C$ by the mean correlation
\begin{equation} \label{eq:def_mean_c}
	c=c(C):=\frac{2}{n(n-1)}\sum_{i>j}C_{ij}		
\end{equation}
and the standard deviation 
\begin{equation} \label{eq:def_variance}
	\sigma=\sigma(C):=\sqrt{\frac{2}{n(n-1)}\sum_{i>j}C^2_{ij}-c^2}
\end{equation} of the correlation coefficients $C_{ij}$. We denote the pair $(c,\sigma)$ as the characteristic of $C$. From \eqref{eq:abs_cij_lower_1},  for any $n \times n$ correlation matrix with $n \geq 2$, one has the constraints
\begin{equation} \label{eq:legal_domain}
|c|,\sigma\leq 1\,
\end{equation}
and 
\begin{equation} \label{eq:circle}
	c^2+\sigma^2 =\frac{2}{n(n-1)}\sum_{i>j}C^2_ {ij} \leq 1.
\end{equation}
 Therefore, the mapping $C \mapsto \left(    c,\sigma   \right)$ maps correlation matrices onto the upper half of the unit disc in the $(c,\sigma)$-plane, as shown in Fig.~\ref{fig:examples} (legal domain).
 In the present paper we address generic relations between the characteristic $(c,\sigma)$ of an arbitrary $n \times n$ correlation matrix $C$  and its spectral structure.  
 \begin{figure}[b]
\includegraphics[width=0.5\textwidth]{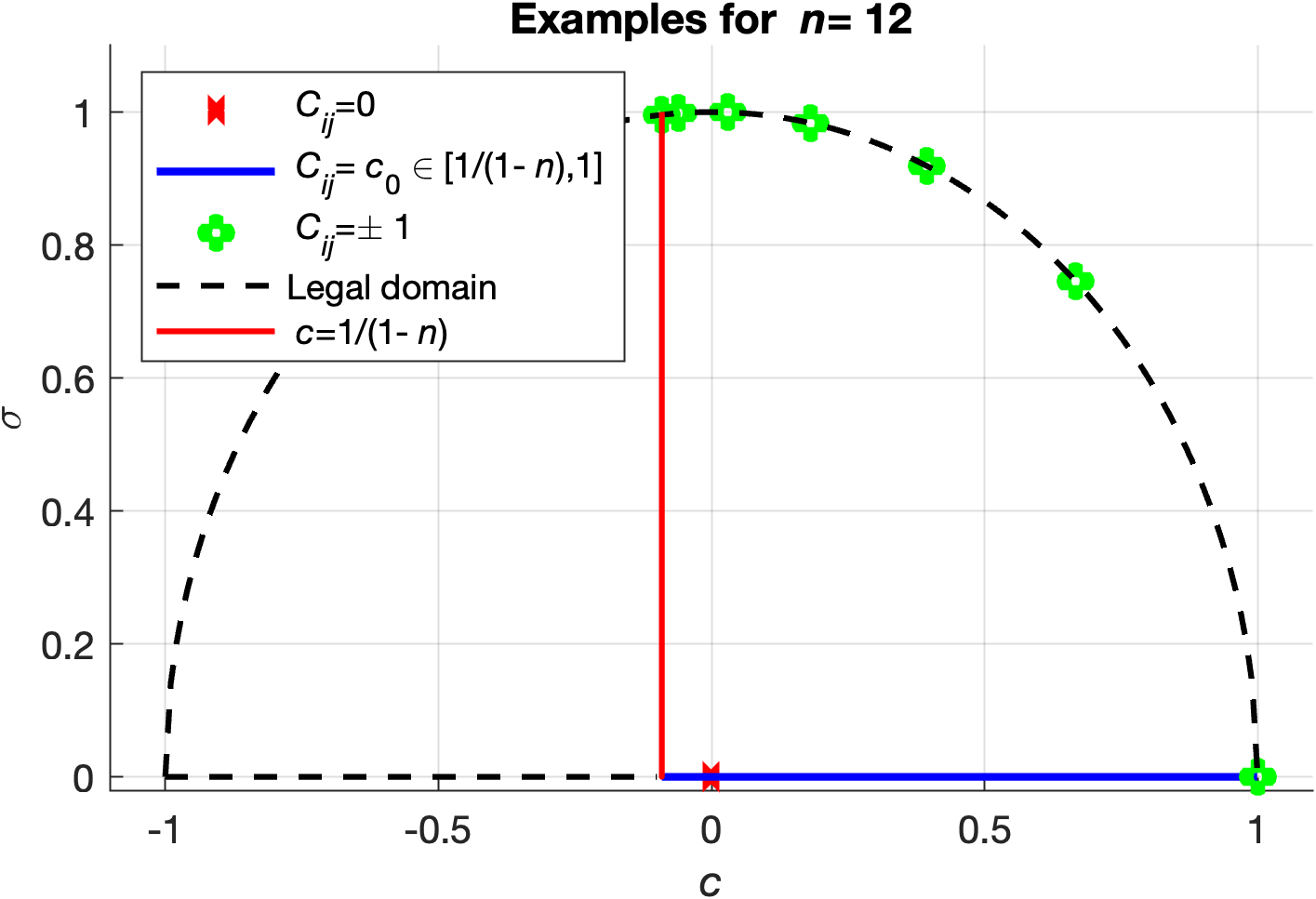}
\caption{The characteristics  $(c,\sigma)$ for correlation matrices form  Examples \eqref{ex:identity} - \eqref{ex:rankOne} for $n=12$.}\label{fig:examples}
\label{fig:examples}
\end{figure}

 Before we consider examples of correlation matrices with  known characteristic $(c,\sigma)$ and their spectral decomposition, we introduce  scalings functions

\begin{equation}\label{eq:DefinitionOfGn}
g_n(x) :=  \frac{(n-1)x+1}{n},
\end{equation}
and
\begin{equation}\label{eq:DefinitionOfSn}
		s(x) := \begin{cases}
			\frac{1}{2}\left(1+\sqrt{2x-1}\right)&, \text{ if } x\geq \frac{1}{2}\\
			x&, \text{ else.}
		\end{cases}
\end{equation} 
We further set
\begin{equation} \label{eq:DefinitionOfS}
	s_n(x) := s(g_n(x)).
\end{equation}
Fig.~\ref{fig:scaling_functions} shows the scaling functions for different $n$ values in the relevant domain. These functions have a distinct meaning for correlation matrices, as we will show. The first observation is the following lemma.
\begin{lemma} \label{lem:sum_and_gn} For a real symmetric $n \times n$ matrix $X$ with $X_{ii} = 1$ for \(1\leq i\leq n\) one has
\begin{equation} \label{eq:gn_sum}
g_n(c(X))=\frac{1}{n^2} \sum_{i,j}X_{ij},
\end{equation} where $c(X)$ is given by   \eqref{eq:def_mean_c}.
	\end{lemma}
	
	\begin{proof}
Using the symmetry $X_{ij}=X_{ji}$  and $X_{ii}=1$ for $1 \leq i,j \leq n$  one has
\begin{equation}
\sum_{i,j}X_{ij}=2\sum_{i < j}X_{ij}+\sum_{ii}X_{ii}=n(n-1)c(X)+n =n^2g_n(c(X)).
\end{equation}
\end{proof} 

 \begin{figure}[t]
 \includegraphics[width=0.4\textwidth]{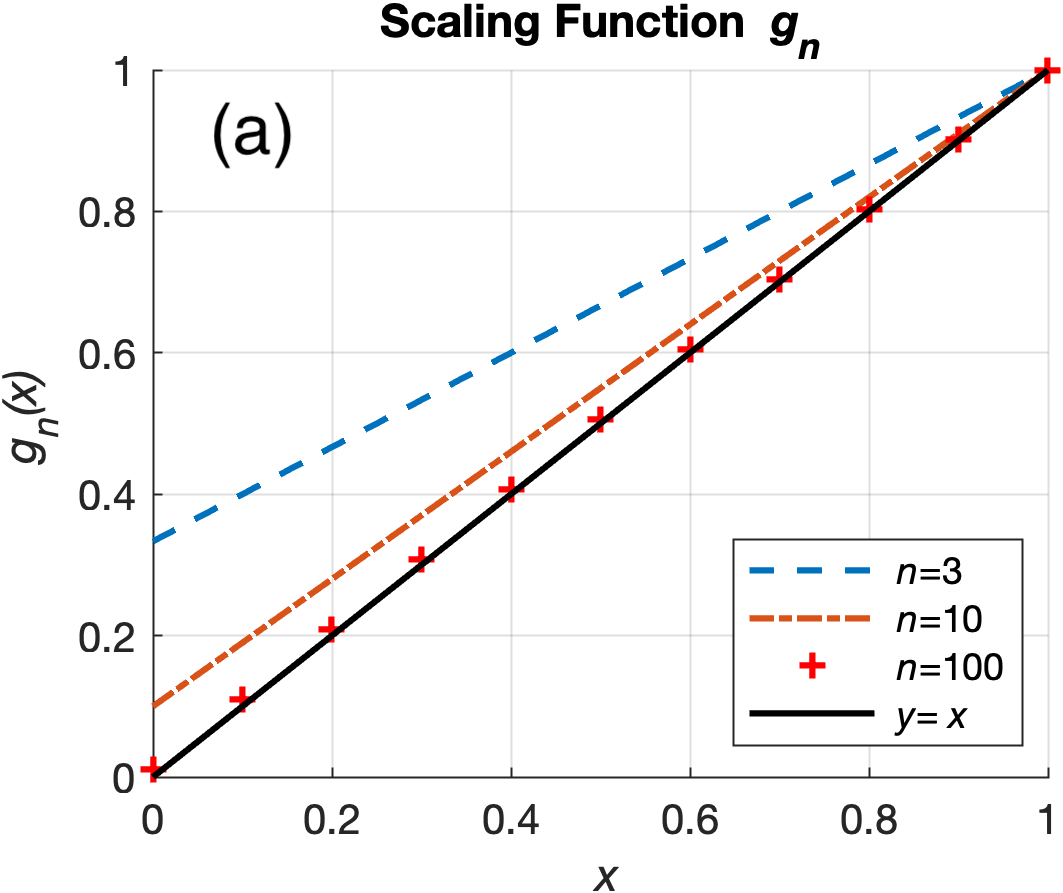}	
\includegraphics[width=0.4\textwidth]{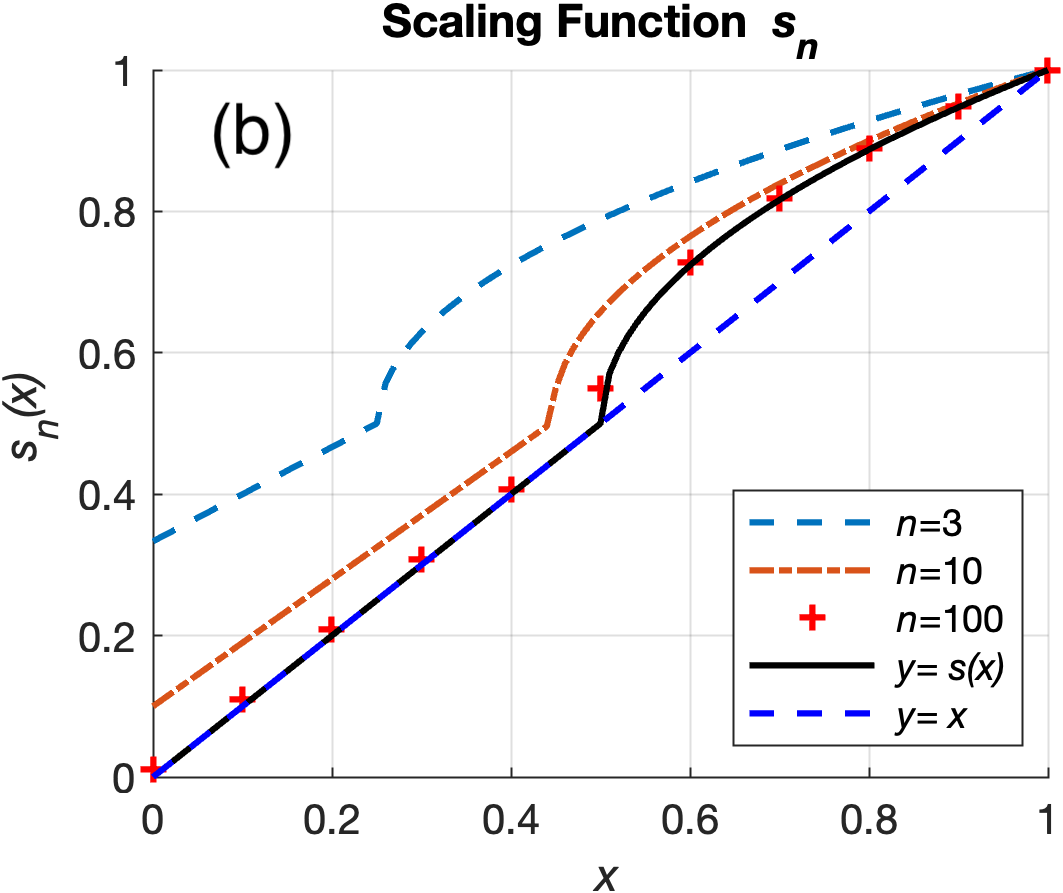}	
\caption{The scaling functions. (a): $g_n$  as defined by  \eqref{eq:DefinitionOfGn}. (b): $s_n$ as defined by \eqref{eq:DefinitionOfS}, for different values of the correlation matrix dimension $n$.}
\label{fig:scaling_functions}
\end{figure}

\subsection{Examples of Correlation Matrices and Motivation from Empirical Observations} \label{sec:examples}

In this section we consider examples of distinct correlation matrices, which we also use for the proofs of our main result. We pick correlation matrices with the characterisic $(c, \sigma)$ from the boundary of the legal domain, given by \eqref{eq:legal_domain} and \eqref{eq:circle}. We especially address the spectral decomposition of correlation matrices,  depending on the position in the $(c, \sigma)$-plane.

\begin{example} \label{ex:identity} (Identity Matrix) 
The simplest correlation matrix is the $n \times n$ matrix C, with $c(C)=\sigma(C)=0$. There exists only one such correlation  matrix for a given $n$, which is the identity matrix denoted by $\operatorname{Id}_n$. It has zero off-diagonal coefficients and 1 on the diagonal.  It has constant eigenvalues
 \begin{equation}
 \lambda_i = 1,
\end{equation} for $1 \leq i \leq n$.  Furthermore, every orthonormal basis of $\R^n$ is an eigenbasis of $\operatorname{Id}_n$. This matrix has the characteristic $(c, \sigma)$ in the origin 	of the $(c, \sigma)$-plane, as shown in Fig.~\ref{fig:examples}.
\end{example}

\begin{example} \label{ex:sigma_0} (Constant Coefficients)
Consider a real symmetric $n \times n$ matrix $C_0$, with $\sigma(C_0)=0$ and 1 on the diagonal. This matrix has constant off-diagonal coefficients
\begin{equation}
  (C_0)_{ij} =c(C_0)= c_0
\end{equation} 
for all \(1\leq i,j\leq n\), \(i\neq j\).
For $c_0>0$, the first eigenvalue\cite{Kaiser_1968,Morrison1976, Friedman1981}
\begin{equation} \label{eq:lam_1_of_C_0}
  \lambda_1=(n-1)c_o+1=ng_n(c_0),
\end{equation}
has the corresponding eigenvector $v_1=\pm \delta_n$, with
\begin{equation} \label{eq:v1_diag}
  \langle v_1,\delta_n \rangle^2=1.
\end{equation}  The remaining eigenvalues
\begin{equation}
  \lambda_j=1-c_0,
\end{equation}  have corresponding eigenvectors $v_j$, with $\langle v_j, \delta_n \rangle=0$. If $c_0$ is in the range
\begin{equation} \label{eq:example_c_0_range}
  \frac{1}{1-n} \leq c_0\leq 1,
\end{equation} then $C_0$  is a correlation matrix.  
For $c_0=0$ one has $C_0=\operatorname{Id}_n$ from the previous example. For $c_0$ in the range given by \eqref{eq:example_c_0_range},  these correlation matrices  continuously cover an interval on the $c$-axis in the $(c,\sigma)$-plane, as shown in Fig.~\ref{fig:examples}. The extreme case $c_0=1$, and hence $\lambda_1=n$, is contained in the family of correlation matrices from the following example.

\end{example}

\begin{example} (Single positive eigenvalue) \label{ex:rankOne}
Consider an $n \times n$ correlation matrix $C$ with
\begin{equation} \label{eq:l1=n}
  \lambda_1=n
\end{equation} and hence $\lambda_j=0$ for $j>1$. From the spectral decomposition \eqref{eq:C_spectral_decomposition} for the diagonal coefficients one has
\begin{equation} \label{eq:C_pm_1_diagonal}
  1=C_{ii}=n(v_{1i})^2
\end{equation} for $1\leq i \leq n$.   From here one has
\begin{equation}
  v_{1i}=\pm \frac{1}{\sqrt{n}}
\end{equation} and hence
\begin{equation} \label{eq:C_ij_pm_1}
C_{ij}= \pm 1
\end{equation} 
for $1 \leq i,j \leq n$. Such correlation matrices are therefore determined by the single  eigenvector $v_1$. For its alignment with respect to $\delta_n$ one has
\begin{equation} \label{eq:v_1_alignment_l1_n}
    \langle v_1, \delta_n \rangle^2=\frac{1}{n}\left(    \sum_{j=1}^nv_{1j}	  \right)^2=\frac{1}{n}\sum_{i,j}\frac{1}{\lambda_1}C_{ij}=g_n(c)
\end{equation} analogous to \eqref{eq:lam_1_of_C_0}.  We used \eqref{eq:l1=n} and \eqref{eq:gn_sum} for $X_{ij}=C_{ij}$ in the last step. Therefore, the alignment of the first eigenvector $v_1$ is determined by the mean correlation $c$. Such correlation matrices can have the first eigenvector parallel as well as perpendicular to the diagonal vector $\delta_n$.
 We note that replacing $v_1$ by $-v_1$ keeps $C§$ unchanged.  Hence for a given  $n$, we easily count correlation matrices satisfying \eqref{eq:l1=n}. From \eqref{eq:v_1_alignment_l1_n}  one has the quantisation  of the mean correlation
\begin{equation} \label{eq:c_quantised}
   g_n(c)=\frac{1}{n}\left(    \sum_{j=1}^nv_{1j}	  \right)^2=\left(    1-\frac{2k}{n}   \right)^2.
\end{equation} Here $k=0,\ldots,\frac{n}{2}$ when $n$ is even 
and $k=0,\ldots,\frac{n-1}{2}$ when \(n\) is odd. 
For the mean correlation one therefore has $c \geq 1/(1-n)$ similar to \eqref{eq:example_c_0_range}.                                                                                  From \eqref{eq:C_ij_pm_1} one has
\begin{equation} \label{eq:on_the_circle}
  c^2+\sigma^2=1.
\end{equation}

In the $(c, \sigma)$-plane, these correlation matrices are characterised by  points on the  unit circe, with $(c,\sigma)$ determined by  \eqref{eq:c_quantised} and \eqref{eq:on_the_circle}, as shown in Fig.~\ref{fig:examples}. As we will show, any $n \times n$ correlation matrix satisfying  \eqref{eq:on_the_circle} has $\lambda_1=n$ and hence a single positive eigenvalue.
\end{example}

 As we mentioned in Sec.~\ref{sec:intro}, empirical\cite{Laloux1999,Plerou1999,Plerou2002,Song2011,Stepanov_2015} as well as simulated\cite{Friedman1981,SornetteMalevergne2004}  correlation matrices 
have been observed to generically share the features 
\begin{equation} \label{eq:features_2}
  \lambda_1 \approx nc
\end{equation}
and 
\begin{equation} \label{eq:features_3}
  \langle v_1,\delta_n \rangle ^2\approx 1
\end{equation} for various values of $c,\sigma$ and $n$. As the examples show, \eqref{eq:features_3} is not a generic feature of a correlation matrices.

It has been shown by Füredi \textit{et al.}\cite{Furedi_Komlos_1981} and Malevergne \textit{et al.}\cite{SornetteMalevergne2004} that correlation matrices with vanishing standard deviation and positive mean correlation automatically have the features \eqref{eq:features_2} and \eqref{eq:features_3} in probability and in the limit $n  \rightarrow \infty$.  In the following sections  we address the spectral structure of an arbitrary $n \times n$ correlation matrix $C$ for $n\geq 2$ with known characteristic $(c,\sigma$).

\subsection{Methods:  A Characteristic Lemma for Correlation Matrices} \label{sec:methods}

In this section we derive a lemma, which give characteristic  constraints for the spectral structure of a correlation matrix $C$ in terms of $c$ and $\sigma$ for any $n \geq 2$. We will deduce our main results applying geometric methods to the  lemma.

 To quantify the eigenbasis geometry we  introduce the weights
\begin{equation} \label{eq:def_of_weights}
	w_j:=\langle v_j,\delta_n\rangle^2
\end{equation} for $1\leq j \leq n$, which  measure the alignment of the eigenvectors with respect to $\delta_n$. The weights $w_j$, therefore, provide a diagonality measure for   the eigenvectors. We note that the weights in general depend on the eigenbasis choice of the underlying correlation matrix $C$.

 With the weights we rewrite the eigenvectors
\begin{equation}
  v_j=\pm\sqrt{w_j}\delta_n+r_j.
\end{equation} Here  $r_j \in \R^n$ is a vector with $\langle r_j,\delta_n\rangle=0$ for $1 \geq j\geq n$. For its magnitude one has 
\begin{equation}
  \| r_j\|^2=1-w_j.
\end{equation}
For correlation matrices with a diagonal eigenvector $v_i=\pm \delta_n$, one has $w_i=1$. From the eigenbasis orthonormality for such correlation matrices one has $w_j =0$  and hence $v_j=r_j$ for $j \ne i$. More generally, one has the normalisation\cite{marcus1992survey,Mandolesi2020}
\begin{equation} \label{eq:weights_norm}
	\sum_{j=1}^nw_j=1.
\end{equation} 
We note that in general,  for the highest weight
\begin{equation}
w_{\max}:=  \max_{n \geq i \geq 1}(w_i),
\end{equation} one has  $ w_{\max}\ne w_1$, as it is the case in Example \ref{ex:rankOne} for $c>0$ small enough. For correlation matrices with $w_{\max}=w_j > w_1$ for some $j \ne 1$ one therefore has
\begin{equation}
  \langle v_1,\delta_n\rangle^2<\langle v_j,\delta_n\rangle^2.
\end{equation}

We additionally introduce the normalised  eigenvalues 
\begin{equation}
\tilde{\lambda}_i:=\frac{\lambda_i}{n}.
\end{equation} From \eqref{eq:trace}  one has the normalisation
\begin{equation} \label{eq:lamd_norm}
  	\sum_{j=1}^n\tilde{\lambda}_j=1,
\end{equation}  analogous  to  \eqref{eq:weights_norm}. Furthermore, for  the weights and for the rescaled eigenvalues one has
\begin{equation}
  0\leq w_i,\,\tilde{\lambda}_i \leq 1,
\end{equation} for $1\leq i \leq n $. For the characteristic lemma, we introduce the vector of the normalised eigenvalues 
\begin{equation}
  \tilde{\lambda}:=(\lambda_1,\ldots,\lambda_n)/n
\end{equation} and the weights vector 
\begin{equation}
 w:=(w_1,\ldots ,w_n)=(\langle v_1,\delta_n\rangle^2,\ldots,\langle v_n,\delta_n\rangle^2)
\end{equation} respectively.  For any $n \times n$ correlation matrix $C$, the alignment of the corresponding vectors $  \tilde{\lambda},w \in \R^n$ is notably restricted by the characteristic $(c,\sigma)$, as we show in the following lemma.

\begin{lemma} (Characteristic Lemma) \label{lem:Characteristic_Lemma}
For $n \geq 2$, let $C$ be an $n \times n$ correlation matrix with the mean correlation $c$ and the standard deviation $\sigma$. For its eigenvalues $\lambda_1,\ldots,\lambda_n$ and an eigenbasis $v_1,\ldots, v_n$, one has

\begin{equation} \label{eq:geom_sclar_prod}
\langle\tilde{\lambda},w\rangle =g_n(c)
\end{equation}
and
\begin{equation} \label{eq:geom_norm}
\|\tilde{\lambda}\|^2=g_n(c^2+\sigma^2).
\end{equation}
	Here \(\tilde{\lambda}:=(\lambda_1,\ldots,\lambda_n)/n\) denotes the normalised   eigenvalues vector and $w:=(w_1,\ldots ,w_n)$ denotes the weights vector, respectively. 
 \end{lemma} 

\begin{proof} 
Applying Lemma~\ref{lem:sum_and_gn} 
with $X_{ij}=C_{ij}$, one has
\begin{equation}
g_n(c)=\frac{1}{n^2}\sum_{i,j}C_{ij}=\frac{1}{n}\delta_n^TC\delta_n=\frac{1}{n}\sum_{i=1}^{n}\lambda_i \left< \delta_n,v _i \right>^2=\sum_{i=1}^{n}\tilde{\lambda}_i w_i.	\end{equation}
Here, in the third step, we use  \eqref{eq:C_spectral_decomposition} which is an identity from linear algebra \cite{LLDines_1943,marcus1992survey}.

For $X_{ij}=C_{ij}^2$ together with \eqref{eq:circle}, one has 
\begin{equation}
g_n(c^2+\sigma^2)=\frac{1}{n^2}\sum_{i,j}C^2_{ij}=\frac{1}{n^2}\text{Tr}(C^2)=\sum_{k=1}^{n}\tilde{\lambda}^2_k.
\end{equation}
In the second step, we use the symmetry of \(C\). In the third step, we use standard properties of the trace for symmetric matrices. The equality in the third step also appears in context of the  Frobenius norm  (resp.~the  Hilber-Schmidt norm) \cite{Horn_1990,golub_1996}.
\end{proof}

\textbf{Methods}  Geometrically, the characteristic  Lemma~\ref{lem:Characteristic_Lemma} shows how for an $n \times n$ correlation matrix $C$, its characteristics $(c,\sigma)$ determines the norm $\|\tilde{\lambda}\|$ and the scalar product $\langle\tilde{\lambda},w\rangle$, which is equal to a weighted average of the normalised eigenvalues. The vectors $\tilde{\lambda}$ and $w$, both have non-negative components and their projection \begin{equation}
  \langle \tilde{\lambda},\delta_n\rangle=\langle w,\delta_n\rangle=\frac{1}{\sqrt{n}}
\end{equation} onto the diagonal vector $\delta_n$
 is constant for any correlation matrix for a fixed $n \geq 2$. Applying geometric methods, we will show that for correlation matrices with right-hand side in \eqref{eq:geom_norm} large enough, the vector $\tilde{\lambda}$ automatically has a distinctly large component. Hence the underlying correlation matrix has a distinctly large eigenvalue $\lambda_1$. Independently, the right-hand side in \eqref{eq:geom_sclar_prod} is large enough then the weights vector $w$ automatically has	 distinctly large first component  $w_1$ and for the first eigenvector one has $v_1 \approx \pm \delta_n$.  

 We conclude this section by showing that \eqref{eq:example_c_0_range} is generically valid for all correlation matrices.
 
 \begin{corollary} \label{cor:smalles_mean_c}
For  $n \times n$ correlation matrices $C$, with the mean correlation $c=c(C)$, one  generically has

\begin{equation}\label{eq:smallest_c}
c(C) \geq \frac{1}{1-n}. 
\end{equation}
\end{corollary}

\begin{proof}
	Follows directly from \eqref{eq:geom_sclar_prod} with  $g_n(c)=	\langle\tilde{\lambda},w\rangle \geq 0$.
\end{proof}

In other words, a real symmetric $n \times n$ matrix  $C$ with 1 on its diagonal and a mean coefficient $c < 1/(1-n)$, necessarily has a negative eigenvalue and is therefore not a correlation matrix. Especially, if rounding up the empirical mean correlation up to the second decimal, the mean correlation can not be negative for $n \geq 202$. Corollary \ref{cor:smalles_mean_c} is especially important for interpretation of empirical correlations.

\section{Main Results}\label{sec:main_reuslts}

\begin{figure}[b]
\includegraphics[scale=.7]{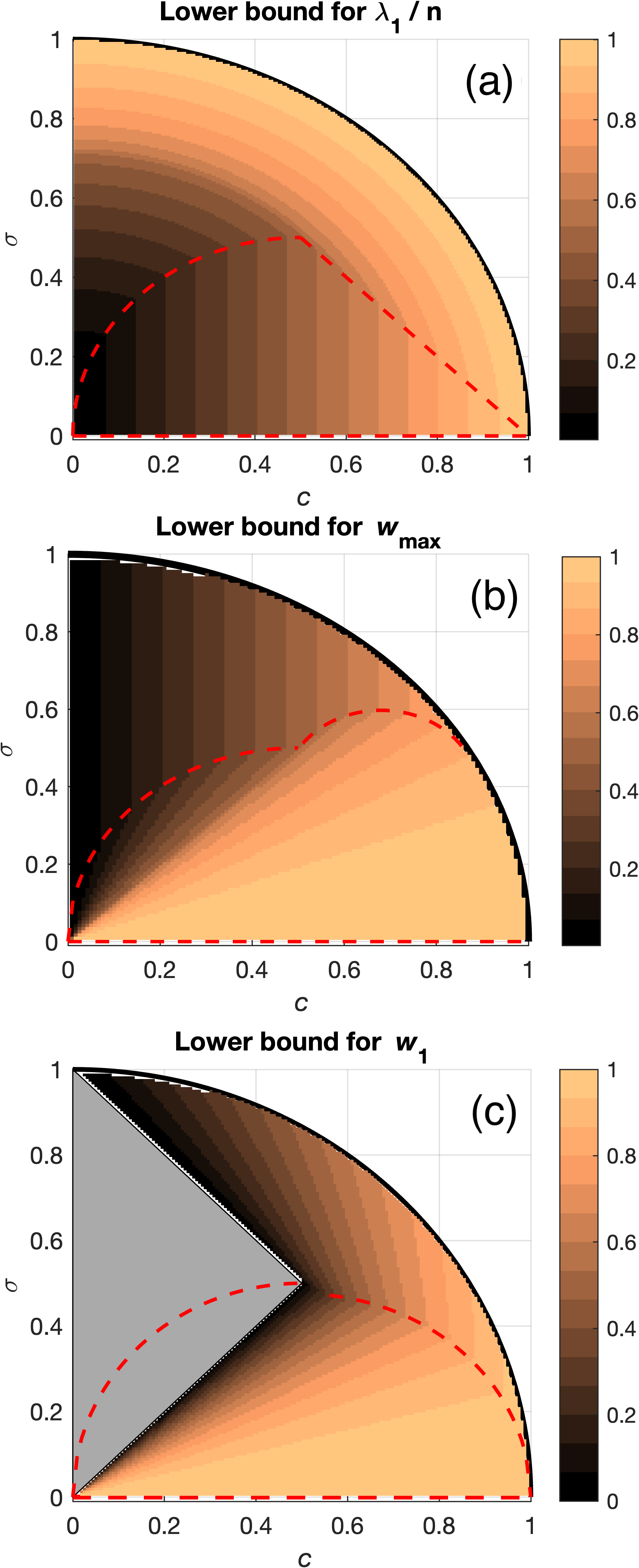}
\caption{The visualisation of the bounds from Corollary~\ref{cor:l_1_w_max} in 15 colour shades from 0 to 1 (colour online).   (a): the lower bound for largest eigenvalue $\lambda_1$ as given by \eqref{eq:l_1_universal}. The dashed curve encloses the domain with $c\geq s(c^2+s^2)$.   (b): the lower bound for $w_{\max}$ as given by \eqref{eq:w_max_universal}. The dashed curve encloses the domain with $c\leq s(c^2/(c^2+\sigma^2))$. (c): the lower bound for $w_1$ as given by \eqref{eq:w_1_universal}. The dashed curve encloses the domain with $\sigma^2/c^2\leq (1-c)/ s(c^2+\sigma^2)$. In the uncoloured area the right-hand side of \eqref{eq:w_1_universal} is negative.}     \label{fig:surfaces}
\end{figure}

In this section we state our main results. We derive lower bounds for the largest eigenvalue \(\lambda_1\), the corresponding weight $w_1$  and the largest weight \(w_\text{max}\) of an $n \times n$ correlation matrix \(C\) in terms of the mean correlation \(c\) and the standard deviation \(\sigma\).

\begin{theorem} \label{th:l_1_w_max}
Let \(C \neq \operatorname{Id}_n\) be an \(n\times n\) correlation matrix, \(n \geq 2\),  with mean correlation \(c=c(C)\), standard deviation  \(\sigma=\sigma(C)\) and let \(\lambda_1 \geq \ldots\geq \lambda_n\) be the eigenvalues of \(C\). Furthermore, let  \(v_1,\ldots v_n \in\R^n\) be an orthonormal basis of \(\R^n\) consisting of eigenvectors of \(C\) with \(Cv_j=\lambda_jv_j\) for \(1\leq j\leq n\). We put \(w_j:=\langle v_j,\delta_n\rangle^2\) for $1 \leq j \leq n$, where \(\delta_n=(1,\ldots,1)/\sqrt{n}\) is the normalised diagonal vector in \(\R^n\). We further put  \(w_{\max}=\max_{1\leq j\leq n}w_j\). We have

\begin{equation}\label{eq:main_result_l_1}
\frac{\lambda_1}{n}\geq \max\{ s_n(c^2+\sigma^2),g_n(c)\}
\end{equation} 
and
\begin{equation}\label{eq:main_result_w_max}
	w_{\max}\geq  \max\{ s_n \left( \frac{c^2}{c^2+\sigma^2} \right),g_n(c)\}.
\end{equation}
		Furthermore, if \(c>0\) then
		\begin{equation}w_1\geq 1-\min\left\{ \frac{n-1}{n}\frac{\sigma^2}{c^2},  \frac{1-g_n(c)}{s_n(c^2+\sigma^2)}\right\}.\end{equation} 
\end{theorem}

Since \(c=\sigma=0\) precisely when \(C=\operatorname{Id}_n\), we need to exclude the identity matrix in Theorem~\ref{th:l_1_w_max} to ensure that the right-hand side of \eqref{eq:main_result_w_max} is well defined.  We note that the estimate \(\lambda_1\geq (n-1)c+1=ng_n(c)\) follows immediately from observations made by Nicewander\cite{NICEWANDER_1974} and Meyer\cite{Meyer_1975} as a consequence of the Min-Max-Principle for eigenvalues. From  \eqref{eq:main_result_l_1} we observe that not only correlation matrices with large  $c$ automatically have a distinctly large eigenvalue. The same is also true for correlation matrices with vanishing mean correlation $c$ and a large standard deviation $\sigma$. 
We state Theorem~\ref{th:l_1_w_max} less sharp but generically valid for correlation matrices of any dimension $n \geq 2$ by replacing $s_n(x)$ by $s(x)$ and $g_n(x)$ by $x$.

\begin{corollary}\label{cor:l_1_w_max}
With the assumptions and the notation as in Theorem \ref{th:l_1_w_max}, one  has
	\begin{equation} \label{eq:l_1_universal}
\frac{\lambda_1}{n}\geq \max\{c,s(c^2+\sigma^2)\} 
\end{equation}
and
\begin{equation} \label{eq:w_max_universal}
w_{\operatorname{max}}\geq \max\left\{c,s\left(\frac{c^2}{c^2+\sigma^2}\right)\right\}.\end{equation}
		If $c>0$ then
\begin{equation} \label{eq:w_1_universal}
w_1\geq 1-\min\left\{ \frac{\sigma^2}{c^2},  \frac{1-c}{s(c^2+\sigma^2)}\right\}.
\end{equation}
	\end{corollary}

The bounds from Corollary~\ref{cor:l_1_w_max} are shown in Fig.~\ref{fig:surfaces}.
From earlier results due to F{\"u}redi--Komlós \cite{Furedi_Komlos_1981}, it follows that for random correlation matrices with $c>0$, one has 
\begin{equation} \label{eq:w_max_fro m_ref}
	w_1                                              \geq 1-4\frac{\sigma^2}{c^2},
\end{equation} in probability when $n$ goes to infinity. For our best knowledge, this result is the only reported bound for $w_1=\langle v_1,\delta_n\rangle^2$ of a correlation matrix in terms of $c$ and $\sigma$.

In the case $w_1=w_{\max}$, the estimates for the corresponding bounds can be improved by taking the maximum of both. We find that this is the case for a wide range of correlation matrices.

\begin{theorem}\label{Thm:w1wmaxDomainNLarge}
		Let \(C\) be an \(n\times n\)- correlation matrix with mean correlation \(c>0\) and standard deviation \(\sigma\). If at least one of the following three conditions
		\begin{itemize}
			\item[(i)] \(\displaystyle c\geq\frac{1}{2}\)
			\item[(ii)] $c \geq \sigma+\frac{1}{\sqrt{n}}$
						\item[(iii)]
			\(\displaystyle c \geq \sqrt[4]{2}\sigma\)
		\end{itemize}
		is satisfied, we have \(w_1>\frac{1}{2}\) and hence \(w_1=w_{\operatorname{max}}\). Here we use the notation as in Theorem \ref{th:l_1_w_max}.
	\end{theorem}

\begin{figure}[b] 
\includegraphics[scale=.9]{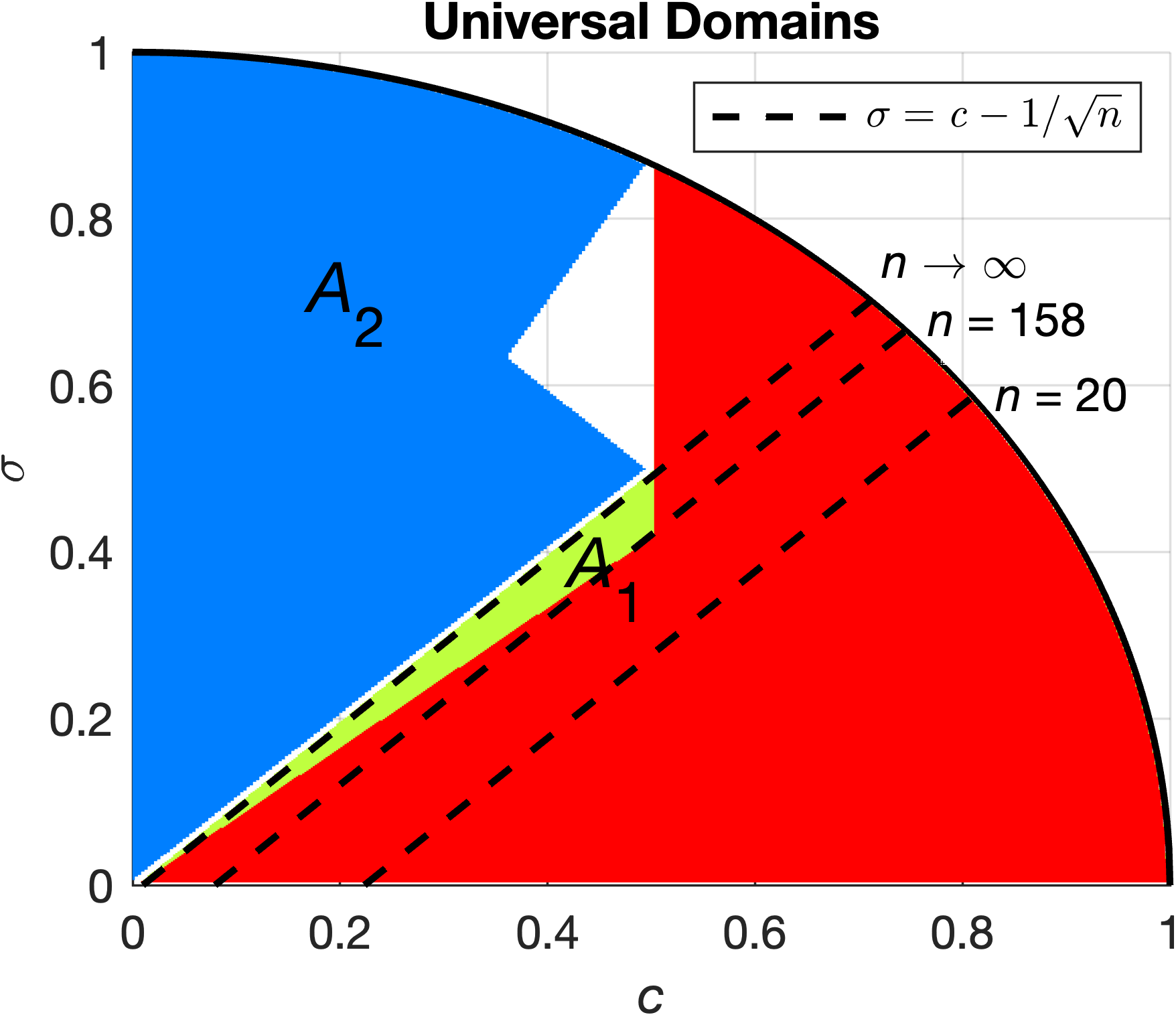}		
\caption{The domains from Theorems~\ref{Thm:w1wmaxDomainNLarge}-\ref{Thm:AsymptoticCounterExamples} (colour online). Red: indicates the domain given by conditions (i) and (iii) in Theorem~\ref{Thm:w1wmaxDomainNLarge}. The area below the dashed lines indicates the domain given by condition (ii) for different values of $n$. The domain $A_1$ defined in  \eqref{eq:A1} is the union of the red and the green areas. The domain $A_2$ defined in \eqref{eq:A2} is indicated by the blue area.} \label{fig:universl_domains}	
\end{figure}

As we already mentioned, it was known\cite{Furedi_Komlos_1981} that correlation matrices with  positive $c$, sufficiently small $\sigma$ and large $n$,  automatically have $w_1=w_{\max}$. Hence Theorem~\ref{Thm:w1wmaxDomainNLarge} extends those results  providing an exact statement on how small \(\sigma\) has to be compared to \(c\) for a  fixed \(n\) . In Fig.~\ref{fig:universl_domains}, the red coloured area shows the domain coming from the inequalities (i) and (iii) which are independent of \(n\). In addition, the domain with $w_1=w_{\max}$ is increased by the \(n\) dependent  inequality (ii) as indicated by the dashed lines in Fig.~\ref{fig:universl_domains}. We note that inequality (ii) becomes relevant  for  $n>158$.

 We now show that $w_1=w_{\max}$ is not a generic feature of correlation matrices and explicetly construct correlation matrices which  satisfy \(w_1<w_\text{max}\).  We will consider correlation matrices with characteristic $(c,\sigma)$ inside the domain \(A:=\{(c,\sigma)\in\R^2\mid c,\sigma>0,c^2+\sigma^2<1\}\) but outside the domain defined in Theorem \ref{Thm:w1wmaxDomainNLarge}. For that purpose, we introduce the following two subdomains of $A$: 
		 \begin{eqnarray} \label{eq:Def_of_domains}
		 A_1&:=&\{(c,\sigma)\in A\mid c\geq \frac{1}{2} \text{ or }\sigma<c \}, \label{eq:A1}	\\ 
		 A_2&:=&\{(c,\sigma)\in A\mid \sigma> \sqrt{3 }c \text{ or }c<\sigma<1-c\}. \label{eq:A2}
		 \end{eqnarray}
The domains $A_1$ and $A_2$ are shown in Fig.~\ref{fig:universl_domains}. 	We find that the domain given by Theorem~\ref{Thm:w1wmaxDomainNLarge}, where \(w_1=w_\text{max}\) holds, approaches the domain \(A_1\) when \(n\to \infty\). Therefore, for any $n \times n$ correlation matrix with \((c,\sigma)\in A_1\) we have $w_1=w_{max}$, provided that $n$ is large enough. We find correlation matrices with characteristic $(c,\sigma) \in A_2$ such that $w_1<w_{\max}$.
	\begin{theorem}\label{Thm:AsymptoticCounterExamples}
		 For any compact set \(K \subset A_2\) there is an \(n_0\in\N\) such that for any \(n\geq n_0\) and any \((c,\sigma)\in K\) there exists a correlation matrix \(C\)  with mean correlation \(c\) and standard deviation \(\sigma\) such that \(w_1<w_{\operatorname{max}}\) holds for any choice of orthonormal eigenbasis. Here, the first eigenvector $v_1$ corresponds to the largest eigenvalue of \(C\).
	\end{theorem}
	
In principle Theorem~\ref{Thm:AsymptoticCounterExamples} shows that for any \((c,\sigma)\in A_2\) and any sufficiently large $n \geq2$ there exist  \(n\times n\) correlation matrices with mean correlation \(c\) and standard deviation \(\sigma\) such that any eigenvector belonging to the largest eigenvalue does not have the largest weight.  The conclusions of Theorem~\ref{Thm:w1wmaxDomainNLarge} and Theorem~\ref{Thm:AsymptoticCounterExamples} do not cover the whole domain $A$. So far we do not know whether correlation matrices with $(c,\sigma)  \in A \setminus (A_1 \cup A_2)$ necessarily  satisfy $w_1=w_{\max}$ or not.
\begin{question} \label{q:question}
 Given $(c,\sigma) \in A \setminus (A_1 \cup A_2)$ and \(n\geq 2\) sufficiently large, is there an \(n\times n\) correlation matrix \(C\) with mean correlation \(c\) and standard deviation \(\sigma\) such that \(w_1<w_{\max}\) holds? 
	\end{question}
	
\textbf{Polar Coordinates:} We round up our results by a remarkable observation. We introduce the polar coordinates 
\begin{equation}
	r_c := \sqrt{c^2+\sigma^2}   \in  (0,1]
\end{equation}
and
\begin{equation}
	 \phi_c  :=       \arccos \left(\frac{c}{r_c} \right) \in [0,\pi]
\end{equation} in the $(c,\sigma)$-plane. From the Theorem~\ref{th:l_1_w_max} for the estimates given by the scaling function $s_n$ one has		
\begin{equation}\label{eq:l_1_spherical}
		\frac{\lambda_1}{n}\geq  s_n(r_c^2),				\end{equation} 
\begin{equation} \label{eq:w_max_spherical}
	w_{\max}\geq s_n(\cos^2 \phi_c).
\end{equation} 
With the polar coordinates, the surface charts in Fig.~\ref{fig:surfaces} collapse to one dimensional functions.  We plot these bounds as functions of $r_c$ and $\phi_c$ in Fig.~\ref{fig:spherical}. We therefore observe that not only  correlation matrices with large $c$, but in general correlation matrices with large $r_c$  automatically have a distinctly large eigenvalue. Analogously, not only correlation matrices with small $\sigma$, but in general correlation matrices with small $\phi_c$ automatically have an approximately diagonal eigenvector. More general we have the following.

\begin{corollary} (Diagonal Cone) \label{cor:diagonal_cone}
With the assumptions and the notation as in Theorem \ref{th:l_1_w_max}, for the smallest possible angle 
\begin{equation}
	\theta_{\text{min}} := \arccos \left( \sqrt{w_{\text{max}}}  \right)
\end{equation} between $\delta_n$ and an  eigenvector of $C$, one has the upper bound

\begin{equation} \label{eq:theta_min}
	\theta_{\text{min}} \leq \arccos \left(\sqrt{s_n(\cos^2 \phi_c)}        \right).
\end{equation} Especially, less sharp, but generically valid for all correlation matrices  one has 
\begin{equation}
	\theta_{\text{min}} \leq \phi_c.
\end{equation}  
\end{corollary}

\begin{proof} [\textbf{Poof:}] Follows directly from  \eqref{eq:w_max_spherical}. 
\end{proof}

\begin{remark}
We note that we do not assume $w_1=w_{\max}$ in Corollary~\ref{cor:diagonal_cone}. From Theorem \ref{Thm:w1wmaxDomainNLarge} it follows that if $\tan(\phi_c)\leq 1/ \sqrt[4]{2}$ then the angle between the two dimensional vector  $(c, \sigma)$ in the $(c, \sigma)$-plane and the $c$-axis, bounds the angle between $ \delta_n$ and an eigenvector belonging to the largest eigenvalue in $\R^n$.
\end{remark}

We complete this section by applying our results to the  correlation matrix from Ref.~\cite{SornetteMalevergne2004}. This matrix has $n=406$ and the characteristic  $(c,\sigma)=(0.14,0.017)$. From Theorem~\ref{th:l_1_w_max} we get for the alignment of the first eigenvector 
\begin{equation}
 | \langle v_1,\delta_n\rangle |=\sqrt{w_1}\geq 0.9853,
\end{equation} which confirms that first eigenvector is  basically given by the diagonal vector\cite{SornetteMalevergne2004}. For the the largest eigenvalue we ensure  the validity of the bound
\begin{equation} \label{eq:sornette_estimation}
  \lambda_1 \geq ng_n(c)=57.7.
\end{equation} This number is very close to the estimate $  \lambda_1 \approx nc=56.84$ from Ref.\cite{SornetteMalevergne2004}.

\begin{figure}[t] 
	\includegraphics[width=0.9\linewidth]{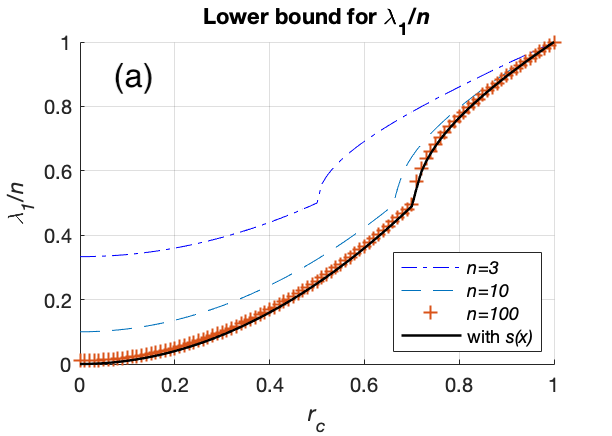}
		\includegraphics[width=.9\linewidth]{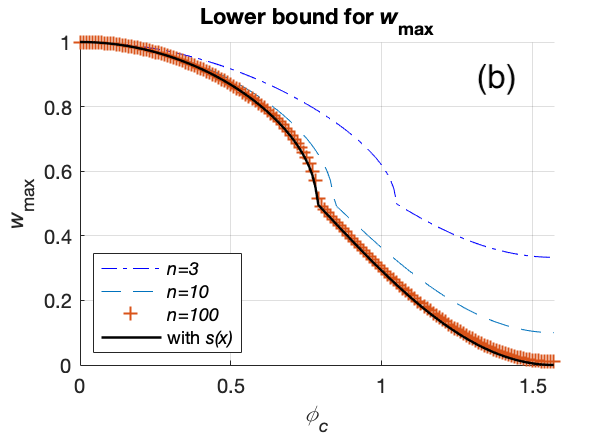}
\caption{(a): the lower bounds for the largest eigenvalue given by \eqref{eq:l_1_spherical}. (b): the bound for the highest weight defined in  \eqref{eq:w_max_spherical} for different $n$ values. We note that one has $c>0$ for $\phi_c<\pi/2$.} \label{fig:spherical}
\end{figure}

\section{Proof of  Main Results: Theorem~\ref{th:l_1_w_max}} \label{sec:proofs_Th1}
	The proof of Theorem~\ref{th:l_1_w_max} will be done in two steps. First we will prove the estimates for \(\lambda_1\) and \(w_{\text{max}}\) in Sec.~\ref{Subsec:l1wmax}. Next we will prove the estimates for \(w_1\) in Sec.~\ref{Subsec:w1Estimates}. Theorem~\ref{th:l_1_w_max} follows in principle from Lemma~\ref{Lem:l1wmaxestimates}, Lemma~\ref{Lem:l1wmaxEstimates}, Corollary~\ref{Cor:w1Estimate}  and Lemma~\ref{Lem:w1EstimatePerturbation}.
	Before we start with the proofs we study some properties of the scaling functions \(g_n\) and \(s_n\).
	Recall the definitions \(g_n(x):=\frac{n-1}{n}x+\frac{1}{n}\)
	and 
	\begin{equation}\label{Eq:DefinitionOfS2}
		s(x)=\begin{cases}
			\frac{1}{2}\left(1+\sqrt{2x-1}\right)&, \text{ if } x\geq \frac{1}{2}\\
			x&, \text{ else.}
		\end{cases}
	\end{equation}
	Furthermore, we put \(s_n(x)=s(g_n(x))\).	We observe the scalings.

\begin{lemma}\label{Lem:gnsnestimate}
For the functions $g_n$, $s_n$ and $s$ 
one has
\begin{equation} \label{eq:scaling_gn}
  g_1(x)\geq  g_{2}(x)\geq \ldots \geq x 
\end{equation}
and
\begin{equation} \label{eq:scaling_sn}
  s_1(x)\geq  s_{2}(x)\geq \ldots \geq s(x),
\end{equation}
for any \(x\in[0,1]\) and \(n\geq 1\). Furthermore, one has \(\lim_{n\to\infty}g_n(x)=x\) uniformly in \(x\) on \([0,1]\). 
\end{lemma}
\begin{proof}

The proof of \eqref{eq:scaling_gn} follows from  $\left(    g_{n}(x)-g_{n+1}(x)   \right)=(1-x)/(n^2+n) \geq 0$ and $g_n(x)=x+(1-x)/n \geq x$. Inequality \eqref{eq:scaling_sn} follows then with the monotony of $s$. Furthermore, we have \(\sup_{x\in[0,1]}|g_n(x)-x|\leq \frac{1}{n}\) which goes to zero when \(n\) tends to infinity.
\end{proof}
 
 Lemma~\ref{Lem:gnsnestimate} immediately verifies how Corollary~\ref{cor:l_1_w_max} follows from Theorem~\ref{th:l_1_w_max}.
	 
	 Throughout this section
	let \(C\neq \text{Id}_n\) be an \(n\times n\) correlation matrix, \(n\geq 2\),  with mean correlation \(c\) and standard deviation  \(\sigma\).  Let \(\lambda_1\geq \ldots\geq \lambda_n\) be the eigenvalues of \(C\) and \(v_1,\ldots v_n \in\R^n\) an orthonormal basis of \(\R^n\) consisting of eigenvectors for \(C\) such that \(Cv_j=\lambda_jv_j\), \(1\leq j\leq n\). For any \(1\leq j\leq n\) put \(w_j:=\langle v_j,\delta_n\rangle^2\) where \(\delta_n=(1,\ldots,1)/\sqrt{n}\) is the normalised diagonal vector in \(\R^n\) and \(w_{\max}=\max_{1\leq j\leq n}w_j\). Furthermore,  we define \(\lambda,w\in\R^n\) by \(\lambda:=(\lambda_1,\ldots,\lambda_n)\) and \(w:=(w_1,\ldots,w_n)\).
%
	\subsection{Estimates for \(\lambda_1\) and \(w_\text{max}\)}\label{Subsec:l1wmax}
		
	\begin{lemma}\label{Lem:l1wmaxestimates}
		We have
		\(\frac{\lambda_1}{n}\geq g_n(c)\) and 
		\(w_{\operatorname{max}}\geq g_n(c)\).
	\end{lemma}
	\begin{proof}
		By Lemma~\ref{lem:Characteristic_Lemma} we have
		\(ng_n(c)=\sum_{j=1}^n\lambda_j w_j\). Using \(\sum_{j=1}^nw_j=1\), \(\sum_{j=1}^n\lambda_j=n\) and \(\lambda_1\geq \lambda_j\geq 0\), \(w_\text{max}\geq w_j\geq 0\) for all \(1\leq j\leq n\) we find
		\begin{equation}g_n(c)\leq \frac{\lambda_1}{n}\sum_{j=1}^nw_j=\frac{\lambda_1}{n}\,\,\,\text{ and }\,\,\,g_n(c)\leq w_\text{max}\sum_{j=1}^n\frac{\lambda_j}{n}=w_\text{max}.\end{equation}
	\end{proof}
	Note that the estimate 
	\(\lambda_1\geq (n-1)c+1\) was already shown by Nicewander\cite{NICEWANDER_1974} and Meyer\cite{Meyer_1975} using the Min-Max principle for eigenvalues.
	
	To obtain the estimates for \(\lambda_1\) and \(w_\text{max}\) in Theorem~\ref{th:l_1_w_max} it remains to show the two inequalities
	\begin{eqnarray}\label{Eq:l1wmaxEstimates}
		\frac{\lambda_1}{n}\geq s_n(c^2+\sigma^2), \,\,\,w_{\operatorname{max}}\geq s_n\left(\frac{c^2}{c^2+\sigma^2}\right).
	\end{eqnarray}
	In order to deduce those estimates from Lemma~\ref{lem:Characteristic_Lemma} we need to estimate the maximum norm of a vector \(p=(p_1,\ldots,p_n)\in\R^n\) with \(0\leq p_j\leq 1\) and \(\sum_{j=1}^{n}p_j=1\) from below when its Euclidean norm \(\|p\|\) is fixed.
	We find the following lower bound. 
	\begin{lemma}\label{Lem:DegenerationEstimate}
		Fix \(n\in\N\) with \(n\geq 2\) and \(\alpha\in[0,1]\).
		Given \(p=(p_1,\ldots,p_n)\in\R^n\) with \(p_j\in[0,1]\) for \(1\leq j\leq n\) and \(\sum_{j=1}^np_j=1\) satisfying \(\|p\|^2\geq \alpha\), we have \begin{equation}\max_{1\leq j\leq n}p_j\geq s(\alpha).\end{equation}
	\end{lemma}
	For the proof of Lemma~\ref{Lem:DegenerationEstimate} the following lemma is crucial.
	\begin{lemma}\label{Lem:FundamentalLemma}
		For any \(n\geq 2\), \(\alpha\in(0,1]\) and \(p=(p_1,\ldots,p_n)\in\R^n\) with \(0\leq p_j<\alpha\) for \(1\leq j\leq n\) and \(\sum_{j=1}^np_j=1\) we have \(\|p\|^2< \alpha\).
	\end{lemma}
	\begin{proof}
		Since \( 0\leq p_j<\alpha\)  we have \(p_j^2\leq \alpha p_j \) for \(1\leq j\leq n\). Put \(p_{\max} =\max_{1\leq j\leq n}p_j\). 
		Since \(\sum_{j=1}^np_j=1\) we have \(0<p_{\max}\) and hence \(p_{\max}^2<\alpha p_{\max}\).
		Without loss of generality we can assume \(p_1=p_{\max}\). Then
		\begin{eqnarray*}
		\|p\|^2&=&p_1^2+\sum_{j=2}^np_j^2\leq p_1^2+\alpha \sum_{j=2}^np_j \\
		&<&\alpha p_1+\alpha \sum_{j=2}^np_j=\alpha\sum_{j=1}^np_j=\alpha.
		\end{eqnarray*}
	\end{proof}
	
	\begin{proof}[\textbf{Proof of Lemma \ref{Lem:DegenerationEstimate}}]
		Without loss of generality assume \(p_1\geq\ldots\geq p_n\) that is \(p_1=\max_{1\leq j\leq n}p_j\). Assuming \(p_1<\alpha\)  immediately leads to \(\|p\|^2<\alpha\) by Lemma~\ref{Lem:FundamentalLemma}. Hence we must have \(p_1\geq \alpha\) which in particular shows \(p_1\geq s(\alpha)\) for \(0\leq \alpha< \frac{1}{2}\). 
		Now assume \(\alpha\geq \frac{1}{2}\).  We have \(p_2+\ldots+p_n=1-p_1\) and hence \(\|p\|^2=p_1^2+p_2^2+\ldots+p_n^2\leq p_1^2+(p_2+\ldots+p_n)^2=p_1^2+(1-p_1)^2\) which implies \(\alpha\leq p_1^2+(1-p_1)^2\). Rearranging this inequality  leads to
		\begin{equation}\frac{1}{4}\left(2\alpha-1\right)\leq \left(p_1-\frac{1}{2}\right)^2.\end{equation}
		Using \(\alpha\geq \frac{1}{2}\) we find
		\begin{equation}\frac{1}{2}\sqrt{2\alpha-1}\leq \left|p_1-\frac{1}{2}\right|\end{equation}
		which is equivalent to
		\begin{equation}p_1\geq \frac{1}{2}+\frac{1}{2}\sqrt{2\alpha-1}\,\, \vee \,\, p_1\leq \frac{1}{2}-\frac{1}{2}\sqrt{2\alpha-1}.\end{equation}
		Since  \(p_1\geq \alpha\geq\frac{1}{2}\) we conclude 
		\begin{equation}p_1\geq \frac{1}{2}+\frac{1}{2}\sqrt{2\alpha-1}=s(\alpha)\end{equation}
		for \(\alpha \geq \frac{1}{2}\).
	\end{proof}

	In order to derive an estimate for the Euclidean norm of the weight vector from Lemma~\ref{lem:Characteristic_Lemma} the following Lemma is needed.
	\begin{lemma}\label{Lem:CorrelationCoefEstimate}
		Let \(\delta_n=(1,1,\ldots,1)/\sqrt{n}\in\R^n\) be the diagonal vector. For any two vectors \(x,y\in\R^n\) we have
		\begin{equation}\left|\langle x,y\rangle-\langle x,\delta_n\rangle\langle y,\delta_n\rangle\right|^2\leq \left(\|x\|^2-|\langle x,\delta_n\rangle|^2\right)\left(\|y\|^2-|\langle y,\delta_n\rangle|^2\right).\end{equation}
	\end{lemma}
	\begin{proof}
		The statement follows immediately from the Cauchy-Bunjakowski-Schwarz inequality applied to the vectors \(\tilde{x}=x-\langle x,\delta_n\rangle\delta_n\) and \(\tilde{y}=y-\langle y,\delta_n\rangle\delta_n\). 
	\end{proof}
	We are now ready to prove the estimates for \(\lambda_1\) and \(w_{\text{max}}\).
	\begin{lemma}\label{Lem:l1wmaxEstimates}
		We have that \eqref{Eq:l1wmaxEstimates} is valid.
	\end{lemma}
	\begin{proof}
		With the notations above we need to prove the two inequalities
		\begin{eqnarray}
			\frac{\lambda_1}{n}&\geq& s_n(c^2+\sigma^2), \label{Eq:MainThmEq1}\\
			w_{\max}&\geq& s_n\left(\frac{c^2}{c^2+\sigma^2}\right). \label{Eq:MainThmEq2}
		\end{eqnarray}
		Put \(\tilde{\lambda}=(\tilde{\lambda}_1,\ldots,\tilde{\lambda}_n):=\frac{1}{n}\lambda\). We have \(\|\tilde{\lambda}\|^2=g_n(c^2+\sigma^2)\) by Lemma~\ref{lem:Characteristic_Lemma} . Since \(0\leq \tilde{\lambda}_j\leq 1\) for \(1\leq j\leq n\) and \(\sum_{j=1}^n\tilde{\lambda}_j=1\) we obtain \begin{equation}\frac{\lambda_1}{n}=\tilde{\lambda}_1\geq s(g_n(c^2+\sigma^2))=s_n(c^2+\sigma^2)\end{equation} from  Lemma~\ref{Lem:DegenerationEstimate} with \(\alpha=g_n(c^2+\sigma^2)\) and \(p=\tilde{\lambda}\). This shows that \eqref{Eq:MainThmEq1} is valid. In order to prove \eqref{Eq:MainThmEq2} we first observe \(\sqrt{n}\langle w,\delta_n\rangle=\sum_{j=1}^n w_j=1\) and \(\sqrt{n}\langle \tilde{\lambda},\delta_n\rangle=\sum_{j=1}^n\tilde{\lambda}_j=1\). Hence we obtain
		\begin{eqnarray}\label{Eq:ApplyCauchySchwarz}
			\left|\langle\tilde{\lambda},w\rangle-\frac{1}{n}\right|^2\leq (\|\tilde{\lambda}\|^2-\frac{1}{n})\left(\|w\|^2-\frac{1}{n}\right),
		\end{eqnarray}
		from Lemma~\ref{Lem:CorrelationCoefEstimate}.
		Using Lemma~\ref{lem:Characteristic_Lemma} we find \(n\langle\tilde{\lambda},w\rangle-1=(n-1)c\) and \(n\|\tilde{\lambda}\|^2-1=(n-1)(c^2+\sigma^2)\).  
		Plugging these identities into \eqref{Eq:ApplyCauchySchwarz} multiplied by \(n^2\) we get
		\begin{eqnarray}
			(n-1)^2c^2\leq(n-1)(c^2+\sigma^2)\left(n\|w\|^2-1\right).
		\end{eqnarray}
		The assumption \(C\neq \text{Id}_N\) ensures \(c^2+\sigma^2\neq 0\). Hence we  obtain after division
		\begin{equation}\|w\|^2\geq g_n\left(\frac{c^2}{c^2+\sigma^2}\right).\end{equation}
		Since \(0\leq w_j\leq 1\) for \(1\leq j\leq n\) and \(\sum_{j=1}^nw_j=1\) we can apply Lemma~\ref{Lem:DegenerationEstimate} with \(\alpha= g_n\left(\frac{c^2}{c^2+\sigma^2}\right)\) and \(p=w\) and find
		\begin{equation}w_{\max}\geq s\left(g_n\left(\frac{c^2}{c^2+\sigma^2}\right)\right)=s_n\left(\frac{c^2}{c^2+\sigma^2}\right).\end{equation}
	\end{proof}
	\begin{remark} \label{rem:symmetric_matr}
		Given an arbitrary symmetric positive semi-definite \(n\times n\)-matrix \(A\neq 0\) with eigenvalues \(\alpha_1\geq\ldots\geq\alpha_n\) we have \(\alpha_1\geq \operatorname{Tr}(A)s\left(\frac{\|A\|^2_F}{\operatorname{Tr}(A)^2}\right)\) where \(\|A\|_F\) denotes the Frobenius norm\cite{Horn_1990,golub_1996} of \(A\). This follows with the same methods as above since \(\alpha_j\geq0\) for \(1\leq j\leq n\), \(\sum_{j=1}^n\alpha_j=\operatorname{Tr}(A)\neq 0\) and \(\|A\|^2_F=\operatorname{Tr}(AA^T)=\operatorname{Tr}(A^2)=\sum_{j=1}^n\alpha_j^2\). We suspect that such an estimate was known before but for our best knowledge we do not know any such studies, except Ref.\cite{TARAZAGA_1990}.
			\end{remark}
	\subsection{Estimates for \(w_1\)}\label{Subsec:w1Estimates}
	\begin{lemma}\label{Lem:w1Estimate}
		Assume \(na\leq\lambda_1\leq nb\) for some numbers $n\geq 2$ and \(a,b\geq0\) such that  \(a+b>1\). Then we have 
		\begin{equation}w_1\geq \frac{g_n(c)+a-1}{b+a-1}.\end{equation}
	\end{lemma}
	\begin{proof}
		By Lemma~\ref{lem:Characteristic_Lemma} we have
		\(ng_n(c)=\sum_{j=1}^n\lambda_j w_j\).
		Hence
		\begin{eqnarray*}
		g_n(c)&=&\frac{\lambda_1}{n}w_1+\frac{1}{n}\sum_{j=2}^n\lambda_j w_j\\ &\leq&\frac{\lambda_1}{n}w_1+(1-\frac{\lambda_1}{n})(1-w_1)\\
		&\leq& bw_1+(1-a)(1-w_1)=w_1(a+b-1)+1-a
		\end{eqnarray*}
		which yields
		\(g_n(c)+a-1\leq w_1(a+b-1)\).
		 Since \(a+b>1\) the claim follows from dividing by \(a+b-1\).
	\end{proof}
\begin{corollary}\label{Cor:w1Estimate}
	We have \(\displaystyle w_1\geq 1-\frac{1-g_n(c)}{s_n(c^2+\sigma^2)}\).
\end{corollary}
\begin{proof}
	Since \(ns_n(c^2+\sigma^2)\leq \lambda_1\leq n\) by Lemma~\ref{Lem:l1wmaxEstimates} and standard properties of correlation matrices the claim follows from Lemma~\ref{Lem:w1Estimate} with \(a=s_n(c^2+\sigma^2)\) and \(b=1\).
\end{proof}
To complete the proof of the estimate for \(w_1\) in Theorem~\ref{th:l_1_w_max} we will use the following lemma from perturbation theory in linear algebra. Note that the following version of that lemma actually follows from a deep result for generalised eigenvalue problems due to Stewart\cite{STEWART197969}. We will give a proof for the simple case we need to consider.
	\begin{lemma}\label{Lem:Wieland}
		Let \(A\) and \(B\) be to symmetric real \(n\times n\)-matrices. Let \(\alpha_1, \ldots,\alpha_n\) be the eigenvalues of \(A\) and \(\xi_1,\ldots,\xi_n\) be an orthonormal basis of corresponding eigenvectors that is \(A\xi_j=\alpha_j\xi_j\), \(1\leq j\leq n\). Given an eigenvalue \(\lambda\) of \(B\) with eigenvector \(v\), \(\|v\|=1\), put \(\eta_k:=\min\{|\lambda-\alpha_j|\mid j\neq k\}\) for every \(1\leq j\leq n\). Then 
		\begin{equation}\eta_k^2(1-\langle v,\xi_k\rangle^2)\leq \|B-A\|^2_F,\text{ for all } 1\leq k\leq n.\end{equation}
	\end{lemma}
	\begin{proof}
		We  rewrite \(v=\sum_{j=1}^na_j\xi_j\) where \(a_j=\langle v,\xi_j\rangle\). It follows
		\begin{equation}\sum_{j=1}^n(\lambda-\alpha_j)a_j\xi_j=Bv-Av= (B-A)v\end{equation}
		Taking the norm of both sides using \(\|(A-B)v\|\leq \|A-B\|_F\|v\|=\|A-B\|_F\), we find
		\begin{equation}\sum_{j=1}^n(\lambda-\alpha_j)^2a_j^2 \leq \|A-B\|^2_F.\end{equation} With
		\begin{equation}(1-a_k^2)\eta_k^2=\sum_{j\neq k}\eta_k^2a_j^2\leq \sum_{j\neq k}(\lambda-\alpha_j)^2a_j^2\leq \sum_{j=1}^n(\lambda-\alpha_j)^2a_j^2\end{equation}
		and \(a_k=\langle v,\xi_k\rangle\) the claim follows.
	\end{proof}
We now apply Lemma~\ref{Lem:Wieland} to the correlation matrix case.
	\begin{lemma}\label{Lem:w1EstimatePerturbation}
		Assume \(c>0\).
		We have \(\displaystyle w_1\geq 1-\frac{n-1}{n}\frac{\sigma^2}{c^2}\).
	\end{lemma}
	\begin{proof}
		Write \(C=C_0+(C-C_0)\) where
		\begin{equation}C_0=
\begin{pmatrix}
 1 & c & \cdots & c\\ 
 c & 1  & \ddots   & \vdots \\ 
 \vdots  & \ddots  & \ddots   & c\\ 
 c & \cdots  & c  & 1
\end{pmatrix}.
\end{equation}
		We know that the eigenvalues \(\alpha_1, \ldots , \alpha_n\) of \(C_0\) are given by \(\alpha_1=ng_n(c)\) and \(\alpha_j=1-c\) for \(j\geq 2\) (see Example~\ref{ex:sigma_0}). Since \(c>0\) we have \(\alpha_1>\alpha_j\) for all \(j\geq 2\). Furthermore, the vector \(\xi_1=\frac{1}{\sqrt{n}}(1,\ldots,1)=\delta_n\) is an normalised eigenvector for the single eigenvalue \(\alpha_1\). From the definition of \(\sigma\) we find \(\|C-C_0\|^2_F=n(n-1)\sigma^2\). Now let \(\lambda_1\) be the largest eigenvalue of \(C\) and \(v_1\) a normalised eigenvector for \(\lambda_1\) that is \(C v_1=\lambda_1 v_1\) and \(\|v_1\|=1\). Applying Lemma~\ref{Lem:Wieland} with \(B=C\) and \(A=C_0\) yields
		\begin{equation}\eta_1^2(1-\langle v_1,\xi_1\rangle^2)\leq n(n-1)\sigma^2\end{equation}
		with \(\eta_1=|\lambda_1-(1-c)|\). Since \(\lambda_1\) is the largest eigenvalue of \(C\) we have \(\lambda_1\geq ng_n(c)=(n-1)c+1>1-c\) by Lemma~\ref{Lem:l1wmaxestimates} and hence \(\eta_1\geq (n-1)c+1-1+c=nc\). It follows
		\begin{equation}1-\langle v_1,\xi_1\rangle^2\leq \frac{n-1}{n}\frac{\sigma^2}{c^2}.\end{equation}
We finish the proof with the observation  \(\langle v_1,\xi_1\rangle^2=\langle v_1,\delta_n\rangle^2=w_1\). 
	\end{proof}
	
\section{Proof of  Main Results: Theorem~\ref{Thm:w1wmaxDomainNLarge}} \label{sec:proofs_T3}
	We will prove Theorem~\ref{Thm:w1wmaxDomainNLarge} in two steps. First we will prove a more general but rather technical version of Theorem~\ref{Thm:w1wmaxDomainNLarge}  which is given as follows.
\begin{theorem}\label{Thm:w1wmaxDomainGeneral}
		Let \(C\neq\operatorname{Id}_n\) be an \(n\times n\)-correlation matrix with mean correlation \(c>0\) and standard deviation \(\sigma\). Let \(\lambda_1\geq\ldots\geq \lambda_n\) be the eigenvalues of \(C\) and \(v_1,\ldots,v_n\) an orthonormal basis of corresponding eigenvectors with weights \(w_j=\langle v_j,\delta_n\rangle^2\), \(1\leq j\leq n\), and \(w_{\operatorname{max}}=\max_{1\leq j\leq n}w_j\). If at least one of the following three conditions
\begin{itemize}
			\item[(i)] \(\displaystyle g_n(c)>\frac{1}{2}\)
			\item[(ii)] \(\displaystyle\frac{g_n(c)^2}{g_n(c^2+\sigma^2)}>\frac{1}{2}\)
\item[(iii)]
			\(\displaystyle s_n\left(\frac{c^2}{c^2+\sigma^2}\right)>\frac{n-1}{n}\frac{\sigma^2}{c^2}\)
\end{itemize}
		is satisfied we have \(w_1>\frac{1}{2}\) and hence \(w_1=w_{\operatorname{max}}\). 
\end{theorem} 

The domain given by (i)-(iii) in Theorem \ref{Thm:w1wmaxDomainGeneral} is shown in Fig.~\ref{Fig:technical_domains} (green+blue). This domain covers a slightly larger area than the domain from Theorem \ref{Thm:w1wmaxDomainNLarge} (blue), especially when  $n$ is small. However, both domains coincide when $n$ goes to infinity.
	
In the second part of the proof we will show that the conditions in Theorem~\ref{Thm:w1wmaxDomainNLarge} imply the conditions in Theorem~\ref{Thm:w1wmaxDomainGeneral}. Theorem~\ref{Thm:w1wmaxDomainGeneral} is in principal the contraposition of Lemma~\ref{Lem:w1wmax12Gen} and Lemma~\ref{Lem:w1wmax3General} below.
	\begin{lemma}\label{Lem:w1wmax12Gen}
		If \(w_1\leq \frac{1}{2}\) we have
		\begin{equation}g_n(c)\leq \frac{1}{2}\,\, \text{ and }\,\, g_n(c)^2\leq \frac{1}{2} g_n(c^2+\sigma^2).\end{equation}
	\end{lemma}
	\begin{proof}
		Since \(\lambda_1\geq\lambda_2\geq\ldots\geq\lambda_n\) we have
		\begin{eqnarray*}
		ng_n(c)&=& \lambda_1w_1+\sum_{j=2}^n\lambda_jw_j\leq\lambda_1w_1+\lambda_2\sum_{j=2}^nw_j\\
		&=&\lambda_1w_1+\lambda_2(1-w_1)=w_1(\lambda_1-\lambda_2)+\lambda_2.
		\end{eqnarray*}
		Since \(\lambda_1\geq\lambda_2\) the assumption \(w_1\leq \frac{1}{2}\) leads  to
		\begin{equation}ng_n(c)\leq \frac{1}{2}(\lambda_1-\lambda_2)+\lambda_2 =\frac{1}{2}(\lambda_1+\lambda_2).\end{equation}
		Using \(\lambda_1+\lambda_2\leq \sum_{j=1}^n\lambda_j=n\) we find \(g_n(c)\leq \frac{1}{2}\). Furthermore, we have
		\begin{equation}g_n(c)^2\leq \frac{1}{4}(\lambda_1/n+\lambda_2/n)^2\leq \frac{1}{2}((\lambda_1/n)^2+(\lambda_2/n)^2)\end{equation}
		where we used \((a+b)^2\leq 2 a^2+2b^2\) for any two real numbers \(a,b\in\R\). Since \begin{equation}(\lambda_1/n)^2+(\lambda_2/n)^2\leq \sum_{j=1}^{n}(\lambda_j/n)^2=g_n(c^2+\sigma^2)\end{equation} we conclude \(g_n(c)^2\leq \frac{1}{2}g_n(c^2+\sigma^2)\).
	\end{proof}
	\begin{lemma}\label{Lem:w1wmax3General}
		Given \(c>0\) and \(w_1\leq \frac{1}{2}\) we have \begin{equation}s_n\left(\frac{c^2}{c^2+\sigma^2}\right)\leq\frac{n-1}{n}\frac{\sigma^2}{c^2}\end{equation}
	\end{lemma}
	\begin{proof}
		Choose \(1\leq k\leq n\) such that \(w_\text{max}=w_k\).
		If \(k\neq1\) we have \(w_1\leq 1-w_k\) since \(\sum_{j=1}^nw_j=1\) and \(0\leq w_j\leq 1\) for \(1\leq j\leq n\) holds. If \(k=1\) we find \(w_1\leq 1-w_k\) since \(w_1\leq \frac{1}{2}\). Hence we conclude that \(w_1\leq 1-w_k=1-w_{\text{max}}\) is satisfied. Using the estimates in Theorem~\ref{th:l_1_w_max} we obtain \begin{equation}1-\frac{n-1}{n}\frac{\sigma^2}{c^2}\leq 1-s_n\left(\frac{c^2}{c^2+\sigma^2}\right).\end{equation}
		Hence the claim follows.
	\end{proof}
We showed that Theorem~\ref{Thm:w1wmaxDomainGeneral} is valid. Then Theorem~\ref{Thm:w1wmaxDomainNLarge}  follows immediately from the three lemmata below.
	\begin{lemma}
		Given \(c\geq\frac{1}{2}\) we have \(g_n(c)> \frac{1}{2}\). 
	\end{lemma}
	\begin{proof}
		For \(c=1\) we have \(g_n(c)=1>\frac{1}{2}\). For \(\frac{1}{2}\leq c<1\) we have
		\(ng_n(c)=(n-1)c+1>nc\geq n/2\).
	\end{proof}
	\begin{lemma}
		Given \(\sigma\leq c-\frac{1}{\sqrt{n}}\) we have \(\frac{g_n(c)^2}{g_n(c^2+\sigma^2)}> \frac{1}{2}\). 
	\end{lemma}
	\begin{proof}
		Since \(\sigma\geq 0\) we can assume \(c\geq \frac{1}{\sqrt{n}}\). which implies
		\begin{equation}\sigma^2\leq \left(c-\frac{1}{\sqrt{n}}\right)^2\leq \left(c-\frac{1}{\sqrt{n}}\right)\left(c+\frac{1}{\sqrt{n}}\right)=c^2-\frac{1}{n}.\end{equation}
		Hence we find
		\begin{eqnarray*}
		ng_n(c^2+\sigma^2)&\leq& (n-1)(2c^2-\frac{1}{n})+1=2nc^2+\frac{1}{n}-2c^2\\
		&\leq& 2nc^2-\frac{1}{n} <2nc^2.
		\end{eqnarray*}
		Since \(g_n(c)\geq c>0\) we conclude
		\begin{equation}\frac{g_n(c)^2}{g_n(c^2+\sigma^2)}>\frac{c^2}{2c^2}=\frac{1}{2}.\end{equation}
	\end{proof}
	\begin{lemma}
		Given \(c>0\) and \(\sigma\leq \frac{1}{\sqrt[4]{2}}c\) we have \begin{equation}s_n\left(\frac{c^2}{c^2+\sigma^2}\right)>\frac{n-1}{n}\frac{\sigma^2}{c^2}.\end{equation} 
	\end{lemma}
	\begin{proof}
		Since \(c>0\) we find \(\sigma\leq c\) and hence \(g_n\left(\frac{c^2}{c^2+\sigma^2}\right)\geq \frac{c^2}{c^2+\sigma^2} \geq \frac{1}{2}\) which implies
		\begin{equation}\label{Eq:DomainCor3Eq1}
		s_n\left(\frac{c^2}{c^2+\sigma^2}\right)=\frac{1}{2}\left(1+\sqrt{2g_n\left(\frac{c^2}{c^2+\sigma^2}\right)-1}\right)\geq \frac{1}{2}.
		\end{equation}
			Assuming that \(\frac{n-1}{n}\frac{\sigma^2}{c^2}< \frac{1}{2}\) immediately leads to the conclusion by \eqref{Eq:DomainCor3Eq1}. So let us assume \(\frac{n-1}{n}\frac{\sigma^2}{c^2}\geq \frac{1}{2}\).
		 From \(\sigma\leq \frac{1}{\sqrt[4]{2}}c\) we obtain \(\sigma^2(c^4-2\sigma^4)\geq 0\). Since
		\begin{eqnarray*}
		\sigma^2(&c^4&-2\sigma^4)\\
		&=& c^6-2\sigma^4c^2-2\sigma^6+2\sigma^2c^4+2\sigma^4c^2-c^6-\sigma^2c^4\\
		&=& c^6-2\sigma^4(c^2+\sigma^2)+2\sigma^2c^2(c^2+\sigma^2)-c^4(c^2\sigma^2)
		\end{eqnarray*}
		and \(c^4(c^2+\sigma^2)>0\) we find
		\begin{equation}0\leq \frac{\sigma^2(c^4-2\sigma^4)}{c^4(c^2+\sigma^2)}=\frac{c^2}{c^2+\sigma^2}-2\frac{\sigma^4}{c^4} +2\frac{\sigma^2}{c^2}-1\end{equation}
		which leads to
		\begin{equation}\label{Eq:DomainCor3Eq2}
		\frac{c^2}{c^2+\sigma^2}\geq 2\frac{\sigma^4}{c^4} -2\frac{\sigma^2}{c^2}+1=2\left(\frac{\sigma^2}{c^2}-\frac{1}{2}\right)^2+\frac{1}{2}.
		\end{equation}
		From \(\frac{n-1}{n}\frac{\sigma^2}{c^2}\geq \frac{1}{2}\) we find \(\sigma>0\) and hence \(\frac{\sigma^2}{c^2}>\frac{n-1}{n}\frac{\sigma^2}{c^2}\geq \frac{1}{2}\). Moreover, we have \(g_n\left(\frac{c^2}{c^2+\sigma^2}\right)\geq\frac{c^2}{c^2+\sigma^2}\). Then \eqref{Eq:DomainCor3Eq2} implies
		\begin{equation}g_n\left(\frac{c^2}{c^2+\sigma^2}\right)>2\left(\frac{n-1}{n}\frac{\sigma^2}{c^2}-\frac{1}{2}\right)^2+\frac{1}{2}\end{equation} which is equivalent to
		\begin{equation}\frac{1}{2}\left(1+\sqrt{2g_n\left(\frac{c^2}{c^2+\sigma^2}\right)-1}\right)>\frac{n-1}{n}\frac{\sigma^2}{c^2}.\end{equation}
		By \eqref{Eq:DomainCor3Eq1} the claim follows.
	\end{proof}
	\section{Proof of  Main Results: Theorem~~\ref{Thm:AsymptoticCounterExamples}} \label{sec:ProofTheorem4}

	\begin{figure}[t] 
\includegraphics[width=.49\linewidth]{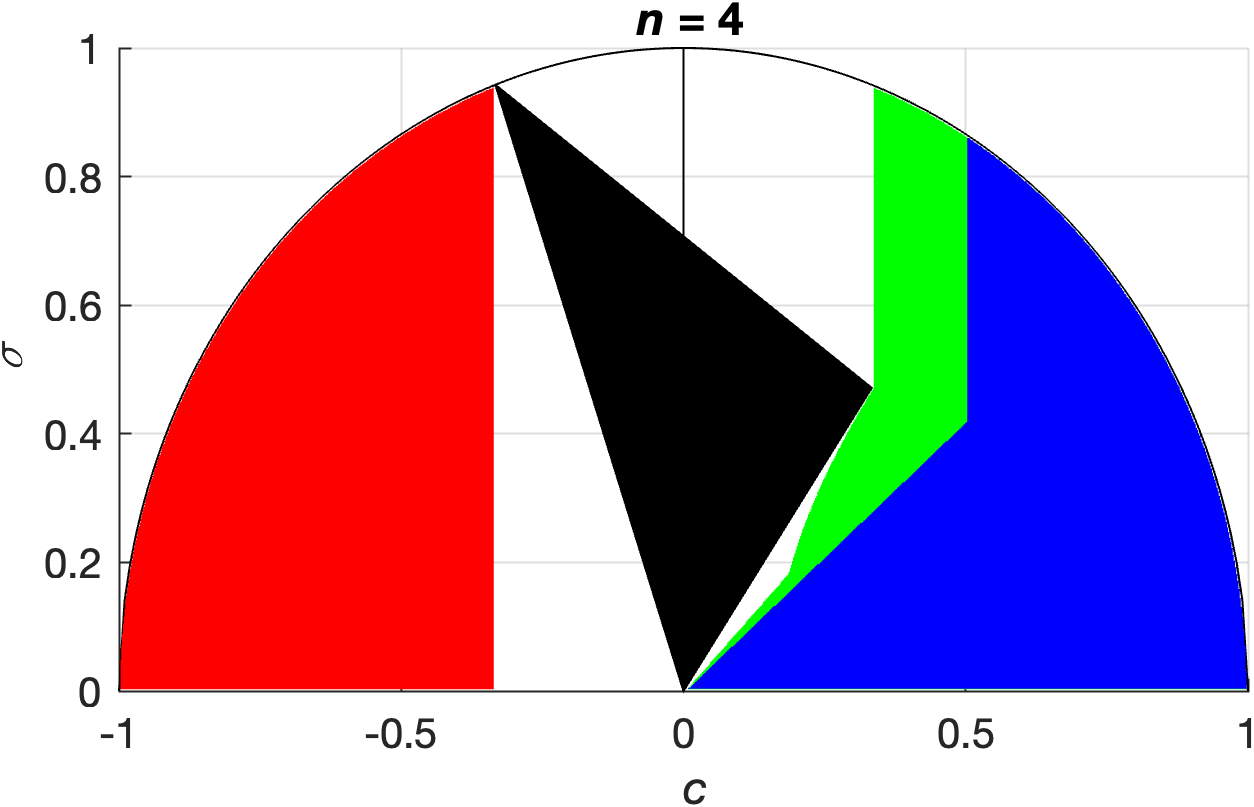}
\includegraphics[width=.49\linewidth]{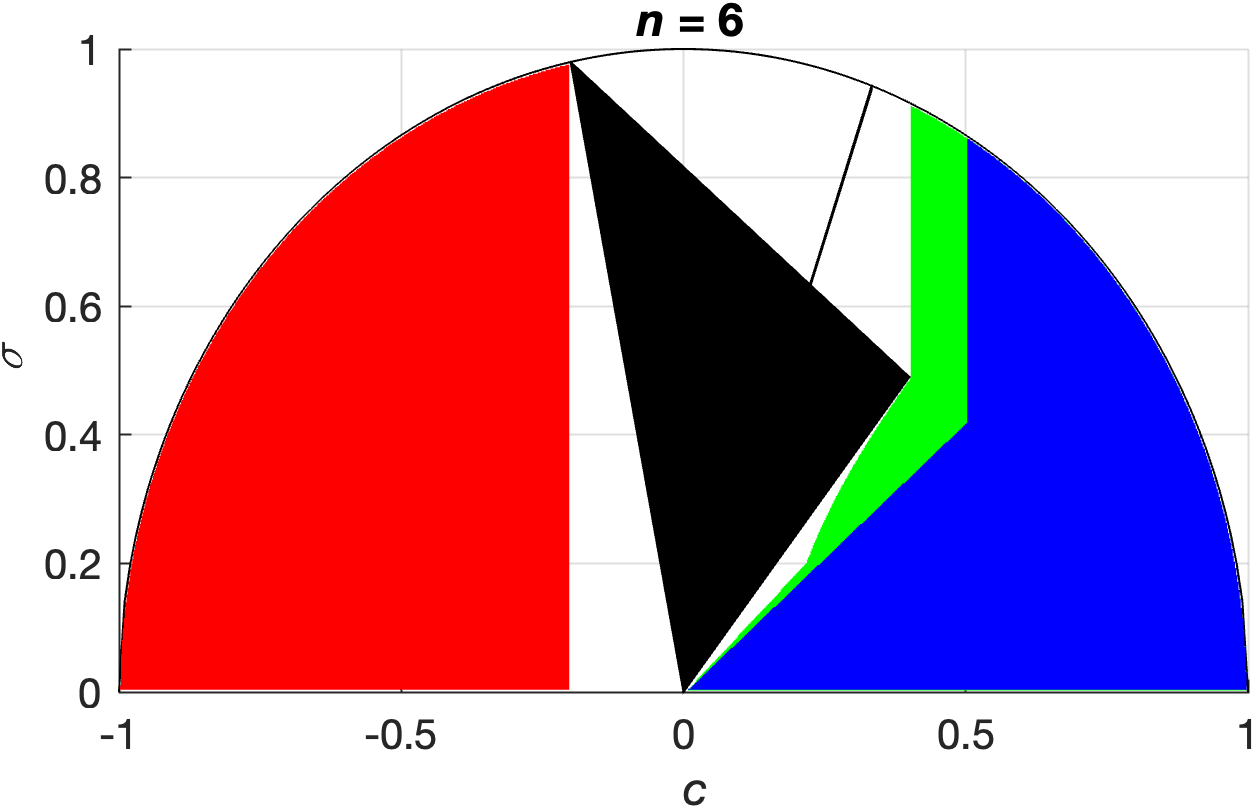}
\includegraphics[width=.49\linewidth]{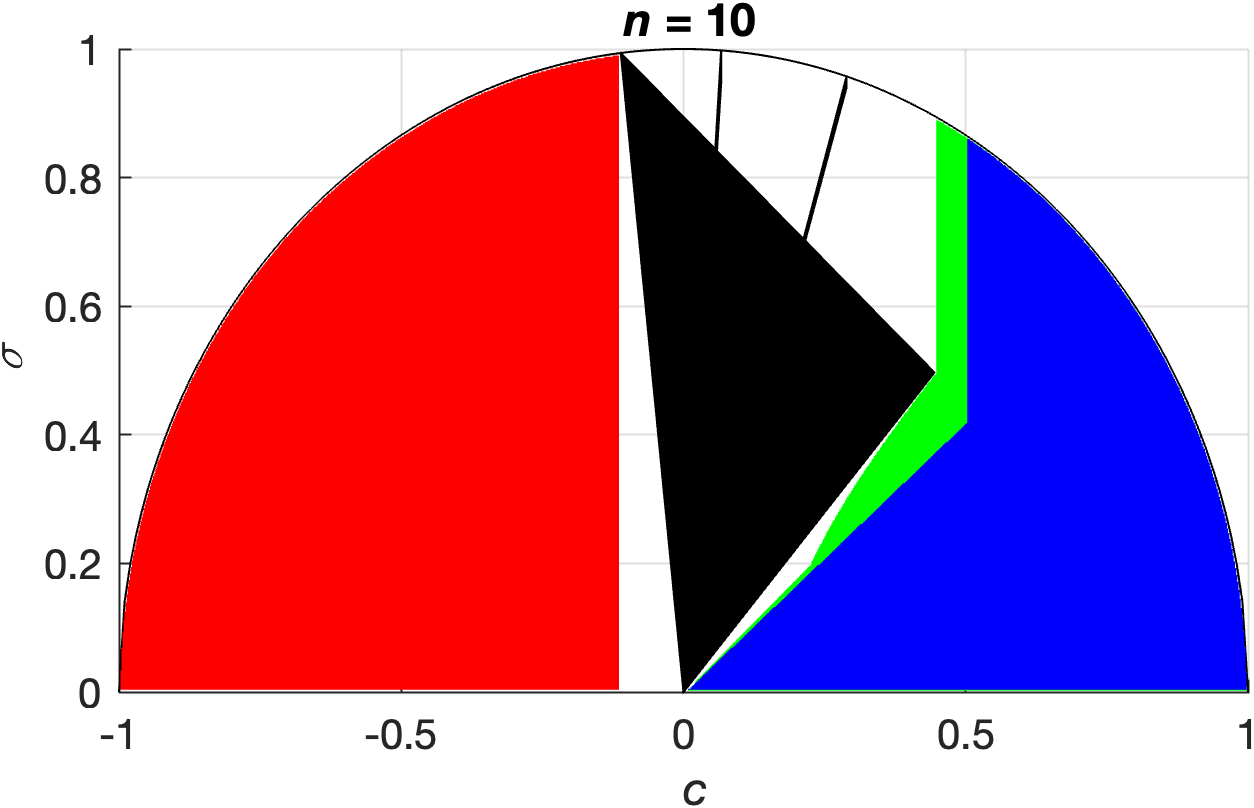}
\includegraphics[width=.49\linewidth]{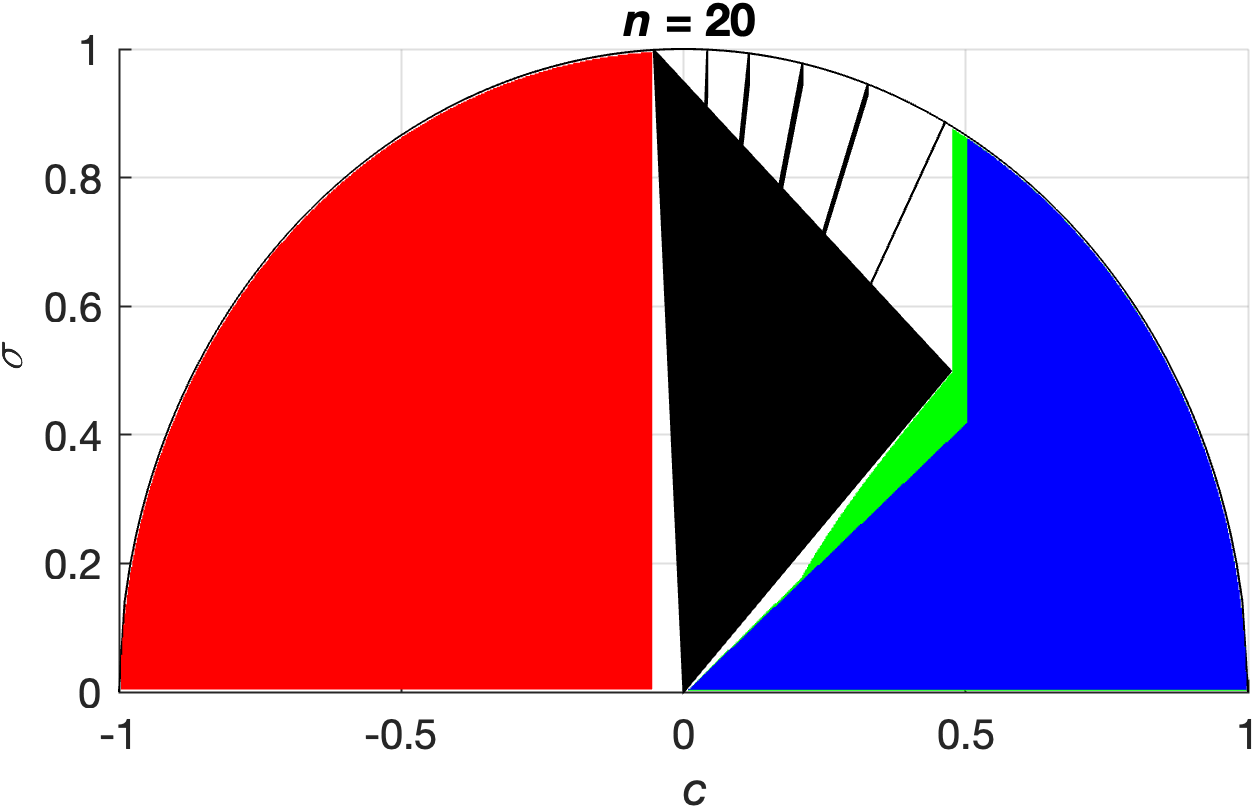}
\includegraphics[width=.49\linewidth]{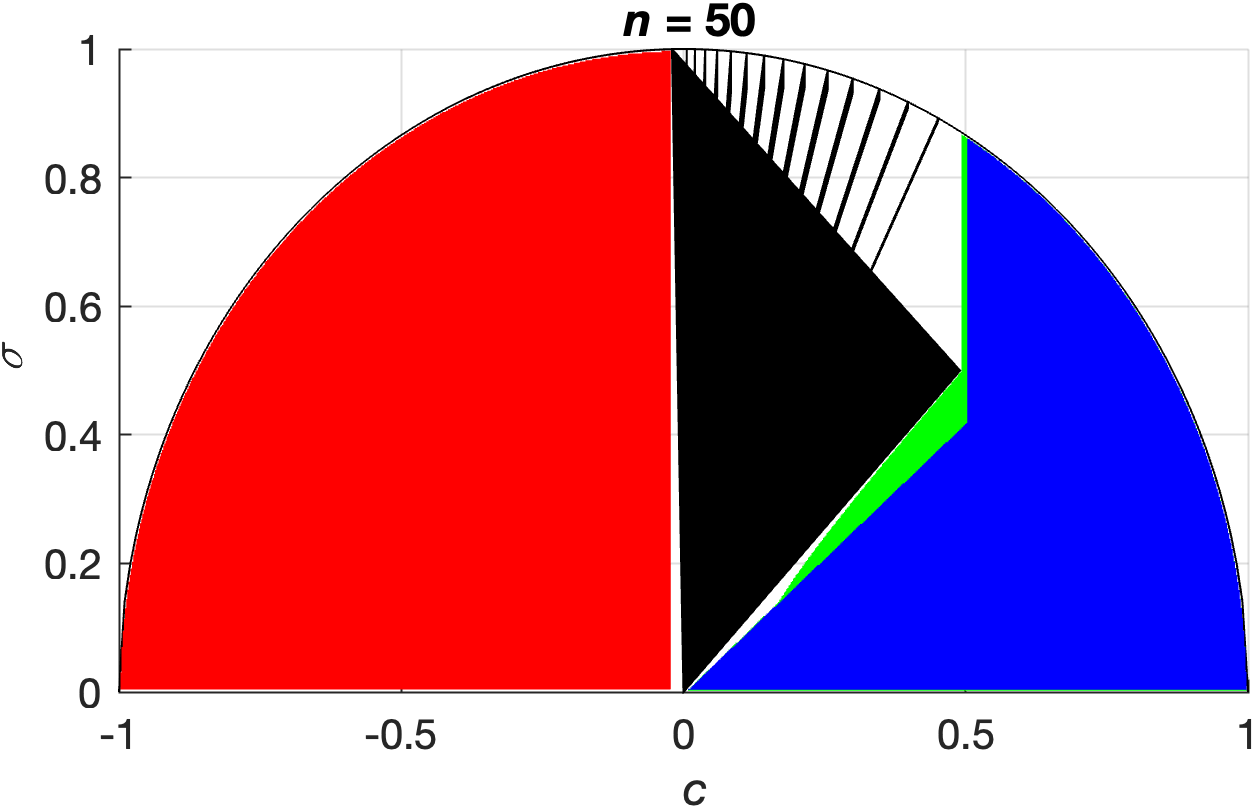}
\includegraphics[width=.49\linewidth]{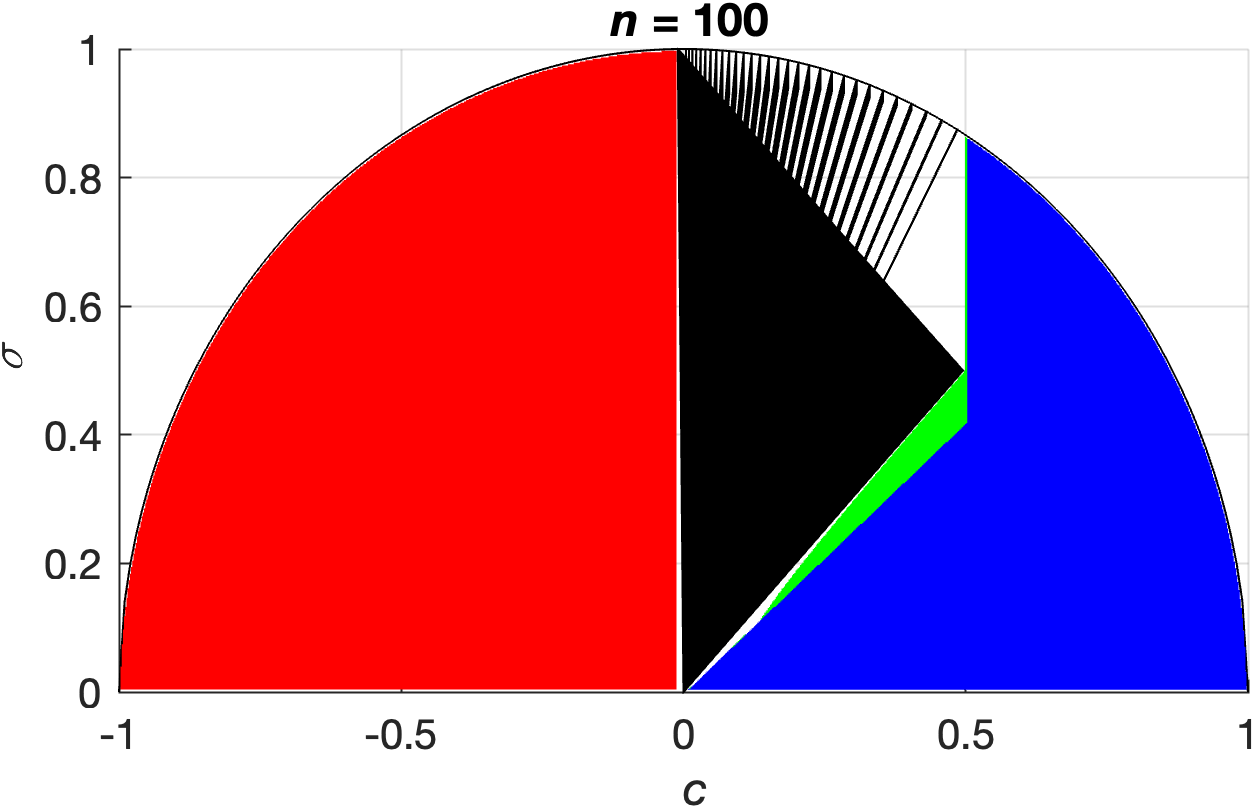}
\includegraphics[width=.49\linewidth]{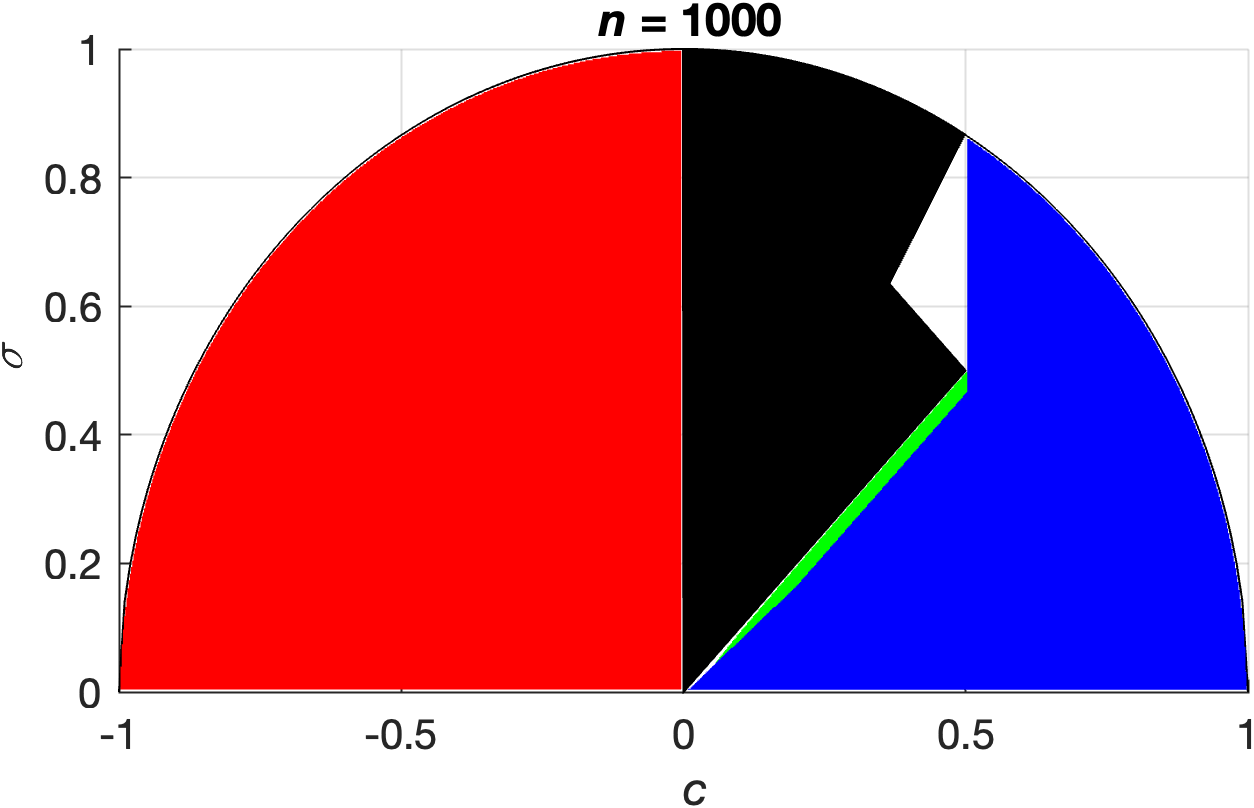}
\includegraphics[width=.49\linewidth]{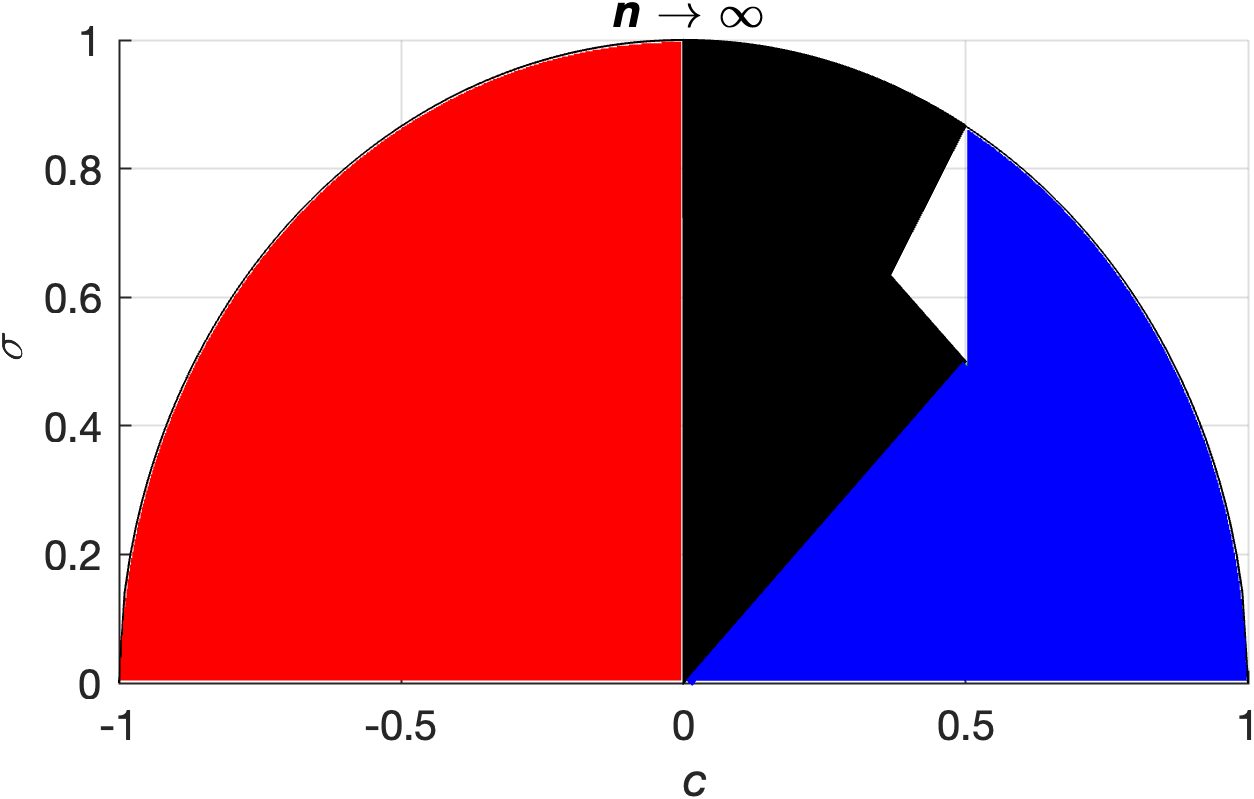}
\caption{The scaling behaviour of the domains from: Corollary~\ref{cor:smalles_mean_c} (red);   Theorem~\ref{Thm:w1wmaxDomainNLarge} (blue); Theorem~\ref{Thm:w1wmaxDomainGeneral} (green+blue); Lemma~\ref{Lem:CounterExampleB1Even} and Proposition~ \ref{Lem:PerturbationOfRankOneCorrelation} (black).
 Red: there are no $n \times n$ correlation matrices with characteristic $(c,\sigma)$ in the red domain. Blue+Green: any $n \times n$ correlation matrices with characteristic $(c,\sigma)$ in the green or blue domain have  $w_1=w_{\max}$. Black: for any $(c,\sigma)$ in the black domain there exists an $n \times n$ correlation matrix with   $w_1<w_{\max}$.} \label{Fig:technical_domains}
\end{figure}

	\begin{definition}
		Let \(C\) be an \(n\times n\) correlation matrix and let  \(\lambda_1,\ldots,\lambda_n\) be its eigenvalues with \(\lambda_1\geq \lambda_j\) for \(1\leq j\leq n\). We say that \(C\) satisfies \(w_1<w_{\text{max}}\) if for any orthonormal basis \(v_1,\ldots,v_n\) such that \(Cv_j=\lambda_jv_j\) and \(w_j:=\langle v_j,\delta_n\rangle^2\), \(1\leq j\leq n\), we have \(w_1<\max_{2\leq j\leq n}w_j\). 
	\end{definition} 
	
	In order to prove Theorem~\ref{Thm:AsymptoticCounterExamples} we write the domain \(A_2\) as the a union \(A_2=B_1\cup B_2\) of domains
	\begin{eqnarray*}
	B_1&=&\{(c,\sigma)\mid c<\sigma<1-c,\,\,c,\sigma>0\}\\
	B_2&=&\{(c,\sigma)\mid \sigma>\sqrt{3}c,\,\,c,\sigma>0,\,\,c^2+\sigma^2<1\}.
	\end{eqnarray*}
We then construct correlation matrices satisfying \(w_1<w_{\text{max}}\) with characteristic \((c,\sigma)\in B_1\) or \((c,\sigma)\in B_2\) separately in Theorem~\ref{Thm:LocalResultB1} and Theorem~\ref{Thm:LocalResultB2}. The final proof of Theorem~\ref{Thm:AsymptoticCounterExamples} can be found at the very end of this section. 
	
	For \(B_1\) we will use tensor products (see Lemma~\ref{Lem:TensorCorrelation} below) and embedding methods (see Lemma~\ref{Lem:EmbeddingMethod}) starting from a well known class of correlation matrices from Example~\ref{ex:sigma_0}. For \(B_2\) we start with the class of rank one correlation matrices as in Example~\ref{ex:rankOne} and find the desired examples using techniques from perturbation theory (see Proposition~\ref{Lem:PerturbationOfRankOneCorrelation}).  In both cases we will abusively use convexity arguments which we are going to explain as next. 
	\begin{lemma}[Convexity Argument]\label{Lem:ConvexityArgument}
		The space of \(n\times n\) correlation matrices is convex in \(\R^{n\times n}\). In particular, given two \(n\times n\) correlation matrices \(A=(a_{ij})\) and \(B=(b_{ij})\) we have that for any \(\mu\in[0,1]\) the matrix \(C(\mu)=(1-\mu)A+\mu B\) is a correlation matrix. If \(c_a\) and \(c_b\) denote the mean correlation of \(A\) and \(B\) respectively we have \(c(\mu)=(1-\mu)c_a+\mu c_b\) where \(c(\mu)\) is the mean correlation of \(C(\mu)\). 
	\end{lemma}
	\begin{proof}
		Fix \(\mu\in[0,1]\) and write \(C(\mu)=(c_{ij})\).
		Since the linear combination of two symmetric matrices is again symmetric we have that \(C(\mu)\) is symmetric. For any \(v\in\R^n\) we have \(\langle v,Av\rangle,\langle v,Bv\rangle\geq0\) and hence
		\begin{equation}\langle v,C(\mu)v\rangle=(1-\mu)\langle v,Av\rangle+\mu\langle v,Bv\rangle\geq 0\end{equation}
		which shows that \(C(\mu)\) is positive semi-definite. Furthermore, we have
		\begin{equation}c_{ii}=(1-\mu)a_{ii}+\mu b_{ii}=1-\mu+\mu=1.\end{equation}
		for any \(1\leq i\leq n\). It follows that \(C(\mu)\) is a correlation matrix. From 
		\begin{eqnarray*}
		n(n-1)c(\mu)&=&\sum_{i\neq j} c_{ij}=(1-\mu)\sum_{i\neq j} a_{ij}+\mu\sum_{i\neq j} b_{ij}\\
		&=&n(n-1)((1-\mu)c_a+\mu c_b)
		\end{eqnarray*} we obtain the second part of the claim.
	\end{proof}
	\begin{corollary}\label{Cor:ConvexityWithId}
		Let \(C\) be an \(n\times n\) correlation matrix with characteristic \((c,\sigma)\), eigenvalues \(\lambda_1\geq\ldots\geq \lambda_n\) and an orthonormal eigenbasis of respective eigenvectors \(v_1,\ldots,v_n\). For any \(\mu\in[0,1]\) we have that \(C(\mu):=(1-\mu)C+\mu\operatorname{Id}_n\) is correlation matrix with characteristic \begin{equation}(c(\mu),\sigma(\mu))=((1-\mu)c,(1-\mu)\sigma).\end{equation} Furthermore, for any \(1\leq j\leq n\) and \(\mu\in[0,1]\) we have  \(C(\mu)v_j=\lambda_j(\mu)v_j\) with \(\lambda_j(\mu)=(1-\mu)\lambda_j+\mu\). 
	\end{corollary}
	\begin{proof}
		 Putting \(B=\text{Id}_n\) in Lemma~\ref{Lem:ConvexityArgument} we find that \(C(\mu)\) is a correlation matrix. Writing \(C(\mu)=(c_{ij}(\mu))\) and \(C=(c_{ij})\) we find \(c_{ij}(\mu)=(1-\mu)c_{ij}\) and hence  \((c(\mu),\sigma(\mu))=((1-\mu)c,(1-\mu)\sigma)\) by the scaling behaviour of mean value and standard deviation. Furthermore, we observe that any eigenvector \(v\) of \(C\) for some eigenvalue \(\lambda\) is an eigenvector of \(C(\mu)\) for the eigenvalue \((1-\mu)\lambda+\mu\).
	\end{proof}
\begin{remark}\label{Rmk:ConvexityWithId}
	One easily checks that if a correlation \(C\) satisfies \(w_1<w_{\text{max}}\) the same holds true for \(C(\mu)\), \(0\leq\mu < 1\), defined in Corollary~\ref{Cor:ConvexityWithId}.
\end{remark}

 As we mentioned we will use tensor products to construct correlation matrices with specific features. We first recall the basic notations and facts from multilinear algebra.
	Let \((V,\langle\cdot,\cdot\rangle_V)\) and \((W,\langle\cdot,\cdot\rangle_W)\) two finite dimensional Hilbert spaces. Then an inner product \(\langle\cdot,\cdot\rangle_{V\otimes W}\) on the tensor product \(V\otimes W\) is defined  by \(\langle v\otimes w,v'\otimes w'\rangle_{V\otimes W}:=\langle v,v'\rangle_{V}\langle w,w'\rangle_{V}\) for \(v,v'\in V\), \(w,w'\in W\). Given to linear maps \(A\colon V\to V\), \(B\colon W\to W\), we denote by \(A\otimes B\colon V\otimes W\to V\otimes W\) the linear map defined by \(A\otimes B(v\otimes w):=Av\otimes Bw\) for \(v\in V\) and \(w\in W\). Given eigenvectors \(v\) of \(A\) and \(w\) of \(B\) with respective eigenvalues \(\alpha\) and \(\beta\) we have that \(v\otimes w\) is an eigenvector of \(A\otimes B\) for the eigenvalue \(\alpha\beta\).
	Let \(e_1,\ldots,e_n\) be the standard basis of \(\R^n\), \(f_1,\ldots,f_m\) be the standard basis of \(\R^m\) and \(d_1,\ldots,d_{nm}\) the standard basis of \(\R^{nm}\).  We identify \(\R^{nm}\) with the tensor product \(\R^n\otimes\R^m\) by putting \(e_{i}\otimes f_{j}\mapsto d_{m(i-1)+j}\) for \(1\leq i\leq n\) and \(1\leq j\leq m\). By fixing a basis on a finite dimensional vector space we have a one to one correspondence between linear maps and matrices. Hence given an \(n\times n\)-matrix \(A=(a_{ij})\) and an \(m\times m\)-matrix \(B=(b_{ij})\) we can identify the tensor product of \(A\) and \(B\) (or more precisely the tensor product of the corresponding linear maps) with an \(nm\times nm\)-matrix  \(C=A\otimes B\) where for \(C=(c_{ij})\) we have
	\(c_{m(i-1)+k,m(j-1)+l}=a_{ij}b_{kl}\). Considering the tensor product of correlation matrices under this notation we have the following. 
	\begin{lemma}\label{Lem:TensorCorrelation}
		Let \(C_1\) be an \(n\times n\)-matrix and \(C_2\) be an \(m\times m\)-matrix. If \(C_1\) and \(C_2\) are correlation matrices then the \(mn\times mn\)-matrix \(C_3:=C_1\otimes C_2\) is a  correlation matrix. Furthermore, if \((c_j,\sigma_j)\), \(j=1,2,3\), is the characteristic of \(C_j\) we have the identities
		\begin{eqnarray*}
		g_{nm}(c_3)&=&g_n(c_1)g_m(c_2),\\ g_{nm}(c_3^2+\sigma_3^2)&=&g_n(c_1^2+\sigma_1^2)g_m(c_2^2+\sigma_2^2).
		\end{eqnarray*}
	\end{lemma}
	\begin{proof}
		Given \(x,x'\in\R^n\) and \(y,y'\in\R^{m}\) it follows from the symmetry of \(C_1\) and \(C_2\) that 
		\begin{eqnarray*}
		\langle x\otimes y, C_1\otimes C_2 x'\otimes y' \rangle&=&\langle x\otimes y, C_1x'\otimes C_2  y' \rangle\\
		&=&\langle x, C_1x'\rangle \langle y, C_2y'\rangle\\
		&=&\langle C_1x, x'\rangle \langle C_2y, y'\rangle\\
		&=&\langle C_1\otimes C_2 x\otimes y,  x'\otimes y' \rangle.
		\end{eqnarray*}
		By linearity it follows that \(C_3\) is symmetric. Writing \(C_j=(c^{(j)})_{lk}\), \(j=1,2,3\), we have
		\begin{equation}c^{(3)}_{m(i-1)+k,m(i-1)+k}=c^{(1)}_{ii}c^{(2)}_{kk}=1.\end{equation} 
		Let \(\alpha_1,\ldots,\geq \alpha_n\geq0\) and \(\beta_1,\ldots,\beta_m\geq 0\) be the eigenvalues of \(C_1\) and \(C_2\) respectively. Given  respective eigenbases \(a_1,\ldots,a_n\) and \(b_1,\ldots,b_m\) we find that \(\{a_i\otimes b_j\}_{1\leq i\leq n,1\leq j\leq m}\) is a basis for \(\R^{nm}\) with \(C_1\otimes C_2 a_i\otimes b_j=C_1a_i\otimes C_2b_j=\alpha_i\beta_ja_i\otimes b_j\) for all \(1\leq i\leq n,1\leq j\leq m\). Hence any eigenvalue of \(C_3\) can be written as \(\alpha_i\beta_j\geq 0\) for some \( 1\leq i\leq n,1\leq j\leq m\) which shows that \(C_3\) is positive semi-definite. We have shown that \(C_3\) is a correlation matrix.
		We have \(\delta_{nm}=\delta_n\otimes\delta_n\) and hence
		\begin{equation}\langle \delta_{nm}, C_3 \delta_{nm}\rangle=\langle \delta_{n}, C_1 \delta_{n}\rangle\langle \delta_{m}, C_2 \delta_{m}\rangle.\end{equation}
		Since \(\langle \delta_{n}, C_1 \delta_{n}\rangle=g_n(c_1)\) and \(\langle \delta_{m}, C_2 \delta_{m}\rangle=g_n(c_2)\)   (see Lemma~\ref{lem:Characteristic_Lemma}) we conclude \(g_{nm}(c_3)=g_n(c_1)g_m(c_2)\). We have \(C_3^2=C_1^2\otimes C_2^2\) and \(\text{Tr}(A\otimes B)=\text{Tr}(A)\text{Tr}( B)\) for arbitrary matrices \(A\) and \(B\) . By Lemma~\ref{lem:Characteristic_Lemma} we find
		\begin{eqnarray*}
		(nm)^2g_{nm}(c_3^2+\sigma_3^2)&=&\text{Tr}(C^2_3)
		=\text{Tr}(C_1^2)\text{Tr}(C_2^2)\\
		&=&n^2g_n(c_1^2+\sigma_1^2)m^2g_m(c_2^2+\sigma_2^2)
		\end{eqnarray*}
	\end{proof}
	Note that given two correlation matrices \(C_1\) and \(C_2\) of size \(n\) and \(m\) respectively we have that the eigenvalues of the \(nm\times nm\) correlation matrix \(C_3=C_1\otimes C_2\) are given by the pairwise products of the eigenvalues of \(C_1\) and \(C_2\). In particular, given an orthonormal eigenbasis \(\{v^{(j)}_k\}\) for \(C_j\), \(j=1,2\), we have that \(\{v^{(1)}_k\otimes v^{(2)}_l \mid 1\leq k\leq n,\,\,1\leq l\leq m\}\) is an orthonormal eigenbasis for \(C_3\). Since  \(\delta_{nm}=\delta_n\otimes \delta_m\) and hence \(\langle\delta_{nm},a\otimes b\rangle=\langle\delta_{n},a\rangle\langle\delta_{m},b\rangle\) for \(a\in \R^n\) and \(b\in\R^m\) we have that the weights for \(C_3\) with respect to that eigenbasis are given by the pairwise products of the weights of \(C_1\) and \(C_2\). Using those techniques  we can construct correlation matrices for large \(n\) with certain properties from well known examples (see Example~\ref{ex:sigma_0}). Therefore, let us consider the following class of correlation matrices.

	\begin{lemma}\label{Lem:CEx1Correlation}
		Let \(n\geq 2\) be even and \(0<\varepsilon<1\). Consider a symmetric \(n\times n\)-matrix \(C\) defined by
		\begin{equation}C=\begin{pmatrix}
		1\!\!\!1 & -\varepsilon 1\!\!\!1\\
		-\varepsilon 1\!\!\!1 & 1\!\!\!1 \\
		\end{pmatrix}\end{equation}
		where \(1\!\!\!1\) is the \(n/2\times n/2\)-matrix with all entries equal to one. Then \(C\) is a correlation matrix with \(w_1<w_{\operatorname{max}}\) such that its characteristic \((c,\sigma)\) satisfies \(g_n(c)=(1-\varepsilon)/2\), \(g_n(\sigma^2+c^2)=(1+\varepsilon^2)/2\) and \(\sigma=\sqrt{\frac{n-2}{n}}(1-c)\). 
	\end{lemma}
	\begin{proof}
		We have \(C=C'\otimes 1\!\!\!1\) with 
		\begin{equation}C'=\begin{pmatrix}
		1&-\varepsilon\\
		-\varepsilon & 1
		\end{pmatrix}.\end{equation}
		By Example~\ref{ex:sigma_0} and Lemma~\ref{Lem:TensorCorrelation} we have that \(C\) is a correlation matrix. Let \(\lambda_1\geq\ldots\geq \lambda_n\) denote the eigenvalues of \(C\) and let \(w_1,\ldots,w_n\), \(w_j:=\langle v_j,\delta_n\rangle^2\) for \(1\leq j\leq n\), be the weights with respect to some orthonormal basis \(v_1,\ldots,v_n\) of respective eigenvectors that is, \(Cv_j=\lambda_j\), \(1\leq j\leq n\). Then it follows from Example~\ref{ex:sigma_0} and the considerations above that \(\lambda_1=n(1+\varepsilon)/2\), \(\lambda_2=n(1-\varepsilon)/2\), \(\lambda_j=0\) for \(j\geq 3\). Furthermore, we have \(v_1=\pm (1,-1)/\sqrt{2}\otimes \delta_{n/2}=\pm (1,\ldots,1,-1,\ldots,-1)/\sqrt{n}\), \(v_2=\pm \delta_2\otimes \delta_{n/2}=\pm\delta_n\) and \(v_j\perp \delta_n\) for \(j\geq 3\). Hence for the weights we obtain \(w_j=0\) for \(j\neq 2\) and \(w_2=1\). This proves \(w_1<w_2=w_{\text{max}}\). From Lemma~\ref{lem:Characteristic_Lemma} we obtain in addition that \(g_n(c)=\lambda_2/n=(1-\varepsilon)/2\) and \(g_n(c^2+\sigma^2)=(\lambda_1/n)^2+(\lambda_2/n)^2=(1+\varepsilon^2)/2\). Hence we have \(c=\frac{n(1-\varepsilon)-2}{2(n-1)}\) and \(c^2+\sigma^2=\frac{n(1+\varepsilon^2)-2}{2(n-1)}\). It follows
\begin{eqnarray*}
		4(n-1)^2\sigma^2&=&2(n-1)n(1+\varepsilon^2)-4(n-1)-((1-\varepsilon)n-2)^2\\
		&=&n(n-2)(1+\varepsilon)^2.
\end{eqnarray*}
		Since \(\varepsilon=1-2g_n(c)\) we find by the definition of \(g_n\) that \(1+\varepsilon=2(n-1)(1-c)/n\). In conclusion we have \begin{equation}\sigma^2=\frac{n-2}{n}(1-c)^2\end{equation} which finishes the proof of the statement.
\end{proof}
Using the convexity argument we can construct a lot of examples for correlation matrices  with even dimension satisfying \(w_1<w_{\text{max}}\) from Lemma~\ref{Lem:CEx1Correlation}.

\begin{lemma}\label{Lem:CounterExampleB1Even}
		Let \(n\geq 4\) be even and \(c,\sigma\) two real numbers such that
\begin{equation} \label{eq:triangle_domain}
\max\left\{-\sqrt{n(n-2)}c,\sqrt{\frac{n}{n-2}}c\right\}< \sigma \leq \sqrt{\frac{n-2}{n}}(1-c)
\end{equation} holds.
Then there exists an \(n\times n\) correlation matrix with characteristic \((c,\sigma)\) 
		 such that \(w_1<w_{\operatorname{max}}\). Furthermore, one can choose such a correlation matrix with no eigenvalue equal to one provided that \(c>0\) holds.  
	\end{lemma}
The triangle defined by \eqref{eq:triangle_domain} is contained back coloured area in Fig.~\ref{Fig:technical_domains}.
	\begin{proof} [Proof of Lemma \ref{Lem:CounterExampleB1Even}]
		Let us first assume that \(\sigma=\sqrt{\frac{n-2}{n}}(1-c)\) holds. 
		In that case we just need to show that \(\varepsilon:=1-2g_n(c)\) satisfies \(\varepsilon\in(0,1)\). Because then the existence of an \(n\times n\) correlation matrix with  characteristic \((c,\sigma)\) and \(w_1<w_{\text{max}}\) follows immediately from Lemma~\ref{Lem:CEx1Correlation}. From the assumptions on \(c\) and \(\sigma\) we find on the one hand
		\begin{equation}\sqrt{\frac{n}{n-2}}c<\sqrt{\frac{n-2}{n}}(1-c)\end{equation}
		which leads to \(nc<(n-2)(1-c)\) and hence \(2(n-1)c+2<n\). Dividing by \(2n\) shows \(g_n(c)<\frac{1}{2}\) which leads to \(\varepsilon>0\).  On the other hand we have
		\begin{equation}-\sqrt{n(n-2)}c<\sqrt{\frac{n-2}{n}}(1-c)\end{equation}
		which leads to \(-ng_n(c)=-(n-1)c-1<0\) and hence to \(g_n(c)>0\). It follows that \(\varepsilon<1\) is valid. Furthermore, for \(c>0\) we find \(\varepsilon<1-\frac{2}{n} \). In that case the eigenvalues \(\lambda_1\geq\ldots\geq \lambda_n\) of the correlation matrix in Lemma~\ref{Lem:CEx1Correlation} satisfy \(\lambda_1> \lambda_2=n(1-\varepsilon)/2>1\) and \(\lambda_j=0<1\) for \(3\leq j\leq n\).  Now let us assume that
		\begin{equation}\label{Eq:TriangleCondition}
		\max\left\{-\sqrt{n(n-2)}c,\sqrt{\frac{n}{n-2}}c\right\}< \sigma < \sqrt{\frac{n-2}{n}}(1-c)
		\end{equation}
		is satisfied. We observe that the condition \eqref{Eq:TriangleCondition} defines an open triangle \(\Delta\) in the \(c,\sigma\)-plane with one vertex at the origin and its  opposite edge \(S\) defined by (see also the considerations above)
		\begin{eqnarray}
		\sigma=\sqrt{\frac{n-2}{n}}(1-c)\,\, \text{ and } 0\leq g_n(c)\leq \frac{1}{2}.
		\end{eqnarray}
		Given a point \((c,\sigma)\in \Delta\) we find that the line through the origin and the point \((c,\sigma)\) intersects \(S\) in a point \((c',\sigma')\). Since \(\Delta\) is open it turns out that \(0<g_n(c')<\frac{1}{2}\) is valid. Then (as shown before) there exists an \(n\times n\) correlation matrix \(C'\) with characteristic \((c',\sigma')\) and \(w_1<w_{\text{max}}\). Putting \(C(\mu)=(1-\mu)C'+\mu\text{Id}\) we find by Corollary~\ref{Cor:ConvexityWithId} and Remark~\ref{Rmk:ConvexityWithId} that there exists a \(\mu\in(0,1)\) such that \(C(\mu)\) has characteristic \((c,\sigma)\) and \(w_1<w_{\text{max}}\). Furthermore, assuming \(c>0\)  leads to \(c'>0\). Hence we can choose \(C'\) such that no eigenvalue is equal to one. Then it follows from Corollary~\ref{Cor:ConvexityWithId} that no eigenvalue of \(C(\mu)\), \(0<\mu<1\), is equal to one which proves the second part of the statement.
	\end{proof}
In order to construct correlation matrices with odd dimension with similar properties as the correlation matrices in Lemma~\ref{Lem:CounterExampleB1Even} we need the following embedding method in combination with Lemma~\ref{Lem:CharacteristicForEmbedding}.
	\begin{lemma}\label{Lem:EmbeddingMethod}
		Let \(C'\) be an \(n\times n\) correlation matrix with characteristic \((c',\sigma')\). Then 
		\begin{equation}C:=\begin{pmatrix}
		C' & 0\\
		0 & 1
		\end{pmatrix}\end{equation}
		is an \((n+1)\times (n+1)\) correlation matrix with characteristic \((c,\sigma)\) such that \(c=\frac{n-1}{n+1}c'\) and \(c^2+\sigma^2=\frac{n-1}{n+1}(c'^2+\sigma'^2)\). Furthermore, if \(C'\) satisfies \(w_1<w_{\operatorname{max}}\) the same holds for \(C\) provided that no eigenvlaue of \(C'\) is equal to one.
	\end{lemma}
	\begin{proof}
		It is obvious that \(C\) is a correlation matrix. Write \(C=(c_{ij})\) and \(C'=(c'_{ij})\). Then \((n+1)nc=\sum_{i\neq j}c_{ij}=\sum_{i\neq j}c'_{ij}=n(n-1)c'\) and \((n+1)n(c^2+\sigma^2)=\sum_{i\neq j}c^2_{ij}=\sum_{i\neq j}c'^2_{ij}=n(n-1)(c'^2+\sigma'^2)\). Now let \(\lambda_1,\ldots , \lambda_{n}\) be the eigenvalues of \(C'\) with \(\lambda_{1}\geq\lambda_{j} \), \(1\leq j\leq n\). By the assumptions on \(C'\) we have \(\lambda_1>\lambda_j\), \(2\leq j\leq n\), and hence by the properties of correlation matrices \(\lambda_1>1\). Putting \(\lambda_{n+1}:=1\) we have by the structure of \(C\)  that \(\lambda_1,\ldots,\lambda_{n+1}\) are the eigenvalues of \(C\) with \(\lambda_1>\lambda_j\), \(2\leq j\leq n+1\). Now let \(v_1,\ldots,v_{n+1}\) be a respective orthonormal eigenbasis and denote by \(w_j=\langle v_j,\delta_{n+1}\rangle^2\), \(1\leq j\leq n+1\) the corresponding weights. By the structure of \(C\) and the assumptions that \(\lambda_j\neq 1\) for \(1\leq j\leq n\) we find  \(v_{n+1}=\pm(0,\ldots,0,1)\) and that the \((n+1)\)-th entry of \(v_j\) is zero for \(1\leq j\leq n\). We denote by \(\tilde{v}_j\in\R^n\) the projection of \(v_j\) onto its first \(n\) components and put \(\tilde{w}_j=\langle \tilde{v}_j,\delta_{n}\rangle^2\) for \(1\leq j\leq n\). Then \(\tilde{v}_1,\ldots,\tilde{v}_n\) is an orthonormal eigenbasis for \(C'\) with respect to its eigenvalues \(\lambda_1,\ldots,\lambda_n\) and hence \(\tilde{w}_1<\max_{1\leq j\leq n}\tilde{w}_j\). Since \(\tilde{w}_j=\frac{n+1}{n}w_j\) we obtain
		\begin{equation}w_1<\max_{1\leq j\leq n}w_j\leq \max_{1\leq j\leq n+1}w_j=:w_{\text{max}}\end{equation}
			which proves the statement.
	\end{proof}
	\begin{lemma}\label{Lem:CharacteristicForEmbedding}
		Given \(n\geq 2\), \(c,c',\sigma,\sigma'\geq 0\) such that \(c=\frac{n-1}{n+1}c'\) and \(c^2+\sigma^2=\frac{n-1}{n+1}(c'^2+\sigma'^2)\) is satisfied. If \(\sigma>0\) we have
		\begin{equation}c'=\frac{n+1}{n-1}c\,\,\,\text{ and }\,\,\,\sigma'=\sigma\sqrt{\frac{n+1}{n-1}}\sqrt{1-\frac{2}{n-1}\frac{c^2}{\sigma^2}}.\end{equation}
	\end{lemma}
	\begin{proof}
	 Since \(n\geq 2\) we immediately observe \(c'=\frac{n+1}{n-1}c\) and \(c'^2+\sigma'^2=\frac{n+1}{n-1}(c^2+\sigma^2)\). Hence we obtain with \(\sigma>0\) that
	 \begin{eqnarray*}
	 \frac{n-1}{n+1}\sigma'^2&=&c^2+\sigma^2 - \frac{n+1}{n-1}c^2\\
	 &=&\sigma^2-\frac{2}{n-1}c^2
	 =\sigma^2\left(1-\frac{2}{n-1}\frac{c^2}{\sigma^2}\right)
	 \end{eqnarray*}
	 which ensures in addition that the right-hand side is non-negative. Then the statement follows from taking the square root on both sides.
	\end{proof}
	Now we are ready to prove the statement of Theorem~\ref{Thm:AsymptoticCounterExamples} in a local formulation for the domain \(B_1\).
	\begin{theorem}\label{Thm:LocalResultB1}
		For any point \(c_0,\sigma_0> 0\) with \(c_0<\sigma_0<1-c_0\) there is  an open neighborhood \(U\) around \((c_0,\sigma_0)\) and \(n_0\in\N\) such that for all
		\((c,\sigma)\in U\) and any \(n\geq n_0\) there exists an \(n\times n\) correlation matrix \(C\) with characteristic \((c,\sigma)\) and \(C\) satisfies \(w_1<w_{\operatorname{max}}\).
	\end{theorem}
	\begin{proof}
		Since \(0<c_0<\sigma_0<1-c_0\) and \(\lim_{n\to\infty}n/(n-2)=1\) we can find \(n'\in\N\) and an open neighborhood \(U'\) around \((c_0,\sigma_0)\) such that \begin{equation}0<\sqrt{\frac{n}{n-2}}c< \sigma < \sqrt{\frac{n-2}{n}}(1-c)\end{equation} holds for all \(n\geq n'\) and all \((c,\sigma)\in U'\). Since \(\lim_{n\to\infty}(n+1)/(n-1)=1\), \(\lim_{n\to\infty}2/(n-1)=0\) and \(\sigma_0>0\)   we can choose \(n_0\geq n'\) and an open neighborhood \(U \subset U'\) around \((c_0,\sigma_0)\) such that
		\begin{equation}\left(\frac{n+1}{n-1}c,\sigma\sqrt{\frac{n+1}{n-1}}\sqrt{1-\frac{2}{n-1}\frac{c^2}{\sigma^2}}\right)\in U'\end{equation}
		for all \(n\geq n_0\) and all \((c,\sigma)\in U\). Then the claim follows from Lemma~\ref{Lem:CounterExampleB1Even}  when \(n\geq n_0\) is even and from Lemma~\ref{Lem:EmbeddingMethod} in combination with Lemma~\ref{Lem:CharacteristicForEmbedding} when \(n\geq n_0\) is odd. 
	\end{proof}

	In order to prove Theorem~\ref{Thm:AsymptoticCounterExamples} we have to show that a similar statement as in Theorem~\ref{Thm:LocalResultB1} also holds for the domain \(B_2\). As mentioned at the beginning of this section we will start with rank one correlation matrices (see also Example~\ref{ex:rankOne}).
	\begin{lemma}\label{Lem:RankOneCorrelation}
		Let \(C\) be an \(n\times n\) correlation matrix with characteristic \((c,\sigma)\). We have  \(c^2+\sigma^2=1\) if and only if \(C\) is of the form \(C=xx^T\) where \(x=(x_1,\ldots,x_n)\) is a vector with \(x_j=\pm1\) for \(1\leq j\leq n\). 
	\end{lemma}
	\begin{proof}
		First assume \(C\) has the form \(C=xx^T\) where \(x=(x_1,\ldots,x_n)\) is a vector with \(x_j=\pm1\) for \(1\leq j\leq n\). Then \(C\) is symmetric and all elements of the diagonal are equal to one. Furthermore, the eigenvalues of \(C\) are \(n\) and \(0\) which shows that \(C\) is positive semi-definite and in addition that \(g_n(c^2+\sigma^2)=1\) by Lemma~\ref{lem:Characteristic_Lemma} which leads to \(c^2+\sigma^2=1\). On the other hand,
		given a correlation matrix \(C\) with \(c^2+\sigma^2=1\) we find by the estimates for \(\lambda_1\) in Theorem~\ref{th:l_1_w_max} that \(n=ns_n(c^2+\sigma^2)\leq \lambda_1\leq n\) which shows \(\lambda_1=n\) and hence \(\lambda_j=0\) for \(j\geq 2\). Hence \(C\) has rank one and from the spectral decomposition we obtain \(C=n\tilde{x}\tilde{x}^T\) where \(\tilde{x}\) is an eigenvector for \(\lambda_1\) of unit length. Writing \(\tilde{x}=(\tilde{x}_1,\ldots,\tilde{x}_n)\) we find since \(C\) is a correlation matrix that \(1=n\tilde{x}_j^2\) has to be satisfied for all \(1\leq j\leq n\). This shows \(\tilde{x}_j=\pm \frac{1}{\sqrt{n}}\) for all \(1\leq j\leq n\). Putting \(x=\sqrt{n}\tilde{x}\) completes the prove. 
	\end{proof}
	\begin{corollary}\label{Cor:RankOneCorrelation}
		Given \(n\geq 2\) and \(0\leq k \leq n\) put \(c=(n(2k/n-1)^2-1)/(n-1)\) and \(\sigma=\sqrt{1-c^2}\). Then \(c^2+\sigma^2=1\) and there exists an \(n\times n\) correlation matrix \(C\) with mean correlation \(c\) and standard deviation \(\sigma\).
	\end{corollary}
	\begin{proof}
		Choose \(x=(x_1,\ldots,x_n)\) with \(x_j=-1\) for \(1\leq j\leq k\) and \(x_j=1\) for \(k< j\leq n\). Then by Lemma~\ref{Lem:RankOneCorrelation} the matrix \(C:=xx^T\) is a correlation matrix with characteristic \((c,\sigma)\) such that \(c^2+\sigma^2=1\) which yields \(\sigma=\sqrt{1-c^2}\). Furthermore, \(g_n(c)=\langle x,\delta_n\rangle^2=(2k/n-1)^2\) by Lemma~\ref{lem:Characteristic_Lemma} which leads to the conclusion.
	\end{proof}
	Now we will use perturbation theory in order to construct correlation matrices satisfying \(w_1<w_{\text{max}}\) from  rank one correlation matrices.
\begin{proposition} \label{Lem:PerturbationOfRankOneCorrelation}

		Let \(C\) be an \(n\times n\)  correlation matrix with characteristic \((c,\sigma)\) such that \(c^2+\sigma^2=1\) and \(1/n< g_n(c)< 1/2\). Given
		\begin{equation}0<\mu<\min\left\{\frac{1}{6}\left(\frac{1}{\sqrt{2}}-\sqrt{g_n(c)}\right),1-\sqrt{\frac{2}{3}}\right\}\end{equation} 
		we have that \(C(\mu):=(1-\mu)C+\mu C_0\) with 
		\begin{equation}C_0= \begin{pmatrix}
 1 & c & \cdots & c\\ 
 c & 1  & \ddots   & \vdots \\ 
 \vdots  & \ddots  & \ddots   & c\\ 
 c & \cdots  & c  & 1
\end{pmatrix}
 \end{equation}
		is a correlation matrix of characteristic \begin{equation}(c(\mu),\sigma(\mu))=(c,(1-\mu)\sigma)\end{equation} which satisfies \(w_1<w_{\operatorname{max}}\). 
		
\end{proposition}
Combining Proposition~\ref{Lem:PerturbationOfRankOneCorrelation} with Corollary~\ref{Cor:RankOneCorrelation} and Corollary~\ref{Cor:ConvexityWithId} leads to a domain  in the $(c,\sigma)$-plane consisting of a bunch of acute triangles, which depend on $n$. These triangles are contained as a part of the black area in Fig.~\ref{Fig:technical_domains}.  For any $(c,\sigma)$ inside these triangles, we find an $n \times n$ correlation matrix $C$ with characteristic $(c,\sigma)$ and $w_1 < w_{\max}$. When $n$ goes to infinity, the union of these triangles  covers the domain $B_2$. For the proof of Proposition~\ref{Lem:PerturbationOfRankOneCorrelation} we will need the two following technical lemma.
\begin{lemma}\label{Lem:Thecnical01}
	Let \(x,\tilde{x}\in\R^n\) be two vectors with \(\|x\|=\|\tilde{x}\|=1\). Put \(w=\langle x,\delta_n\rangle^2\) and \(\tilde{w}=\langle\tilde{x},\delta_n\rangle^2\). We have
	\begin{equation}1-\langle x,\tilde{x}\rangle^2\geq \frac{3}{4}(\sqrt{w}-\sqrt{\tilde{w}})^2\end{equation}
\end{lemma}
\begin{proof}
	We will first show
	\(|\langle x,\tilde{x}\rangle|\leq \sqrt{w}\sqrt{\tilde{w}}+\sqrt{1-w}\sqrt{1-\tilde{w}}.\) 
	Without loss of generality we can assume \(\langle x,\delta_n\rangle,\langle \tilde{x},\delta_n\rangle\geq0\). Then one has \(x=\sqrt{w} \delta_n+p\),   \(\tilde{x}=\sqrt{\tilde{w}} \delta_n+\tilde{p}\) with \(p,\tilde{p}\perp\delta_n\).
	From \(\|x\|=\|\tilde{x}\|=1\) we obtain \(\|p\|^2=1-w\) and \(\|\tilde{p}\|^2=1-\tilde{w}\). Hence one finds
	\begin{eqnarray*}
	|\langle x,\tilde{x}\rangle|&\leq&|\sqrt{w}\sqrt{\tilde{w}}+\langle p,\tilde{p}\rangle|\\
	&\leq& \sqrt{w}\sqrt{\tilde{w}}+\|p\|\|\tilde{p}\|=\sqrt{w}\sqrt{\tilde{w}}+\sqrt{1-w}\sqrt{1-\tilde{w}}.
	\end{eqnarray*}
	Now since \(\sqrt{1-w}\sqrt{1-\tilde{w}}\leq \frac{1}{2}(2-w-\tilde{w})\) and \begin{equation}1-\frac{1}{2}(\sqrt{w}-\sqrt{\tilde{w}})^2=\sqrt{w}\sqrt{\tilde{w}}+\frac{1}{2}(2-w-\tilde{w})\end{equation} we observe
	\begin{eqnarray*}
		1-\langle x, \tilde{x}\rangle^2 &\geq& 1-(\sqrt{w}\sqrt{\tilde{w}}+\sqrt{1-w}\sqrt{1-\tilde{w}})^2\\
		&\geq& 1-(1-\frac{1}{2}(\sqrt{w}-\sqrt{\tilde{w}})^2)^2\\
		&=& (\sqrt{w}-\sqrt{\tilde{w}})^2(1-\frac{1}{4}(\sqrt{w}-\sqrt{\tilde{w}})^2).
	\end{eqnarray*}
	From \(|\sqrt{w}-\sqrt{\tilde{w}}|\leq 1\) we conclude 
	\begin{equation}1-\langle x, \tilde{x}\rangle^2\geq \frac{3}{4}(\sqrt{w}-\sqrt{\tilde{w}})^2.\end{equation}
\end{proof}
\begin{lemma}\label{Lem:Thecnical02}
	Let \(x\in\R^n\) be a vector with \(\|x\|=1\) and \(0\leq\langle x,\delta_n \rangle < 1\). Put \(w=\langle x,\delta_n\rangle^2\) and
	\begin{equation}\tilde{x}=-\frac{\sqrt{w}}{\sqrt{1-w}}x+\frac{1}{\sqrt{1-w}}\delta_n.\end{equation}
	Then we have \(\tilde{x}\perp x\), \(\|\tilde{x}\|=1\) and \(\tilde{w}:=\langle \tilde{x},\delta_n\rangle^2=1-w\).
\end{lemma}
\begin{proof}
	First we observe that since \(\langle x,\delta_n \rangle<1\) we have \(1-w>0\).
	Hence we write
	\begin{equation}\sqrt{1-w}\langle x,\tilde{x} \rangle=-\sqrt{w}\|x\|^2+\langle x,\delta_n \rangle=-\sqrt{w}+\sqrt{w}=0\end{equation}
	to show that \(\langle x,\tilde{x} \rangle=0\). Then we check
	\begin{equation}\sqrt{1-w}\langle \tilde{x},\delta_n\rangle=-\sqrt{w}\sqrt{w}+1=1-w\end{equation}
	which shows \(\tilde{w}=1-w\). Furthermore, we deduce from \(\langle \tilde{x},\delta_n\rangle=\sqrt{1-w}>0\) that
	\begin{equation}\sqrt{1-w}\langle \tilde{x},\tilde{x} \rangle=-\sqrt{w}\langle x,\tilde{x} \rangle+\langle \delta_n,\tilde{x} \rangle=\sqrt{1-w}\end{equation} which proves \(\|\tilde{x}\|=1\).
\end{proof}
	\begin{proof}[\textbf{Proof of Proposition~\ref{Lem:PerturbationOfRankOneCorrelation}}]
		By Lemma~\ref{Lem:ConvexityArgument} and Example~\ref{ex:sigma_0}  we have that \(C(\mu)\) is a correlation matrix for any \(0\leq \mu\leq 1\). Since \(c^2+\sigma^2=1\) we have by Lemma~\ref{Lem:RankOneCorrelation} that \(C=nv'v'^T\) with \(\|v'\|=1\) and \(|v'_{j}|=1\), \(1\leq j\leq n\), where \(v'=(v'_{1},\ldots,v'_{n})\). Furthermore, we have \(g_n(c)=w':=\langle v',\delta_n\rangle^2\) by Lemma~\ref{lem:Characteristic_Lemma}. After possibly replacing \(v'\) by \(-v'\) we can ensure that \(\langle v',\delta_n\rangle=\sqrt{g_n(c)}\geq0\). Put 
		\begin{equation}\tilde{v}:=-\frac{\sqrt{w'}}{\sqrt{1-w'}}v'+\frac{1}{\sqrt{1-w'}}\delta_n.\end{equation}
		Then we have \(\|\tilde{v}\|=1\) and \(\langle v',\tilde{v}\rangle=0\) by Lemma~\ref{Lem:Thecnical02}.
		Choose \(v_3,\ldots, v_n\) such that \(v',\tilde{v},v_3,\ldots,v_n\) is an orthonormal basis. Since \(\delta_n=\sqrt{1-w'}\tilde{v}+\sqrt{w'}v'\) we find \(\langle v_j,\delta_n\rangle=0\) for \(j\geq 3\).
		Hence \(v_3,\ldots,v_n\) are eigenvectors of \(C(\mu)\) for the eigenvalue \(\mu(1-c)\) (see Example~\ref{ex:sigma_0}). Put \(\rho:=C-C_0\). Then \(C=C_0+\rho\) and with \(\rho=(\rho_{ij})\) we have \(\rho_{ii}=0\), \(1\leq i\leq n\), \(\sum_{1\leq i,j\leq n}\rho_{ij}=0\) and \(\sum_{1\leq i,j\leq n}\rho_{ij}^2=n(n-1)\sigma^2\). Hence we obtain \(c(\mu)=c\) and \(\sigma(\mu)=(1-\mu)\sigma\). Furthermore, we have \(\|C-C(\mu)\|^2_F=\mu^2\|\rho\|^2_F=\mu^2\sigma^2n(n-1)\) where \(\|\cdot\|_F\) denotes the Forbenius norm. Now fix \(\mu\) satisfying the assumptions and let \(v_1\) be an eigenvector of \(C(\mu)\) with respect to its largest eigenvalue \(\lambda_1\) such that \(\|v_1\|=1\) and \(\langle v_1, \delta_n\rangle\geq0\). Put \(w_1:=\langle v_1, \delta_n\rangle^2\). We have \(0<\mu<1/2\) and since \(\lambda_1/n\geq g_n(c_\mu^2+\sigma_\mu^2)\geq c^2+(1-\mu)^2\sigma^2\geq \frac{1}{4}\sigma^2 \), \(Cv_j=0\) for \(j\geq3\), \(C\tilde{v}=0\) we have by Stewart (see Lemma~\ref{Lem:Wieland}) that
		\begin{equation}1-\langle v',v_1\rangle^2\leq\frac{n-1}{n}\frac{\mu^2\sigma^2}{(\frac{1}{4}\sigma^2)^2}\leq \frac{16\mu^2}{\sigma^2}\end{equation}
		and hence by Lemma~\ref{Lem:Thecnical01} that \(\frac{3}{4}(\sqrt{w_1}-\sqrt{w'})^2\leq \frac{16\mu^2}{\sigma^2}\) holds which implies
		\begin{equation}\sqrt{w_1}\leq \frac{8\mu}{\sqrt{3}\sigma}+\sqrt{w'}\leq 6\mu+\sqrt{w'}\end{equation}
		where for the last estimate we use that \(c\leq g_n(c)< 1/2 \) and \(c^2+\sigma^2=1\) implies \(\sigma\geq\frac{\sqrt{3}}{2}\). It follows that given \(\mu<\frac{1}{6}(\frac{1}{\sqrt{2}}-\sqrt{w'})\) we have \(w_1<1/2\). Furthermore, from \(\sigma\geq\frac{\sqrt{3}}{2}\) and the estimate for \(\lambda_1\) above we obtain with \(\mu<1-\sqrt{\frac{2}{3}}\) that \(\lambda_1/n\geq (1-\mu)^2\sigma^2>\frac{1}{2}\). Hence \(\lambda_1>\mu(1-c)\) which implies \(v_1\perp\text{span}\{v_3,\ldots,v_n\}\) by the standard properties of symmetric matrices. Then choose \(v_2 \perp\text{span}\{v_1,v_3,\ldots,v_n\}\), \(\|v_2\|=1\). It follows that \(v_1,v_2,v_3,\ldots,v_n\) is an orthonormal eigenbasis of \(C(\mu)\) with \(Cv_1=\lambda_1v_1\), \(Cv_2=\lambda_2v_2\) and \(Cv_j=\mu(1-c)v_j\) for \(j\geq 3\). Put \(w_j=\langle v_j, \delta_n\rangle^2\), \(j\geq 2\), and \(w_{\text{max}}=\max_{1\leq j\leq n}w_j\). Since \( w_1<1/2\), \(w_j=0\), \(j\geq 3\), and \(w_1+\ldots +w_n=1\) we have \(w_2=1-w_1>1/2\). This shows \(w_1<w_{\text{max}}\) for the specific basis \(v_1,\ldots,v_n\). In order to show that \(w_1<w_{\text{max}}\) holds for any orthonormal eigenbasis of \(C(\mu)\) it lasts out to prove that \(\lambda_2\neq \mu(1-c)\). Because given any  orthonormal eigenbasis \(\tilde{v}_1,\ldots,\tilde{v}_n\) with respect to the eigenvalues \(\lambda_1,\lambda_2\) and \(\mu(1-c)\) of \(C_\mu\) and \(\lambda_2\neq \mu(1-c)\) we immediately find \(\tilde{v}_1=\pm v_1\), \(\tilde{v}_2=\pm v_2\) and \(\langle \tilde{v}_j,\delta_n\rangle=0\) which shows that with \(\tilde{w}_j:=\langle \tilde{v}_j,\delta_n\rangle\) we have \(\tilde{w}_1=w_1<w_2=\tilde{w}_2\).  So let us show that \(\lambda_2\neq \mu(1-c)\) is satisfied. Assuming that \(\lambda_2= \mu(1-c)\) leads to \((1-\mu)Cv_2=-\mu c1\!\!\!1v_2\) where \(1\!\!\!1\) is the \(n\times n\)-matrix with all entries equal to one. Since \(\mu c>0\) this is only possible if \(Cv_2=1\!\!\!1 v_2=0\). But this implies \(v_2\perp \delta_n\) and hence \(v_1=\pm \delta_n\). But this is not possible since \(w_1<1/2\). In conclusion we obtain \(\lambda_2\neq \mu(1-c)\) which shows that \(w_1<w_{\text{max}}\) for any orthonormal eigenbasis of \(C(\mu)\).
	\end{proof}
We wish to cover the domain \(B_2\) by correlation matrices with \(w_1<w_{\text{max}}\) by applying the convexity argument (Corollary~\ref{Cor:ConvexityWithId}) to the matrices in Proposition~\ref{Lem:PerturbationOfRankOneCorrelation}. By Lemma~\ref{Lem:RankOneCorrelation} and Corollary~\ref{Cor:RankOneCorrelation} we have for \(n\) fixed that there exists only finitely many correlation matrices with characteristic \((c,\sigma)\) satisfying \(c^2+\sigma^2=1\). Hence we need to ensure that there exist enough of them when \(n\) is large.
	\begin{lemma}\label{Lem:Thecnical03}
		Let \(0<a<b<1\) two real numbers. There exists \(n_0\in\N\) such that for any \(n\geq n_0\) there is a \(k\in\N\), \(0\leq k\leq n\), with
		\begin{equation}\label{eq:pigeonprinciple}
		    a<\frac{n(2k/n-1)^2-1}{n-1}<b.
		\end{equation}
	\end{lemma}
	\begin{proof}
		Put \(a'=g_n(a)\), \(b'=g_n(b)\). Then \(0<a'<b'<1\) and \eqref{eq:pigeonprinciple} is equivalent to
		\(a'<(2k/n-1)^2<b'\).
		Choose \(n_0\) such that \((\sqrt{b'}-\sqrt{a'})/2>1/n_0\). 
		We find that for any \(n\geq n_0\) there exists \(1\leq k\leq n-1\) with \((\sqrt{a'}+1)/2\in[(k-1)/n,k/n)\) and hence
		\begin{equation}\frac{\sqrt{a'}+1}{2}<\frac{k}{n}<\frac{\sqrt{b'}+1}{2}\end{equation} which is equivalent to \(a'<(2k/n-1)^2<b'\).
	\end{proof}
Now we are ready to prove a result as in Theorem~\ref{Thm:LocalResultB1} for the domain \(B_2\).
	\begin{theorem}\label{Thm:LocalResultB2}
		Let \(c_0,\sigma_0>0\) be two real numbers with \(c_0^2+\sigma_0^2<1\) and \(\sigma>\sqrt{3}c\). Then there in an open neighborhood \(U\) around \((c_0,\sigma_0)\) and a positive integer \(n_0\) such that for any \((c,\sigma)\in U\) and any \(n\geq n_0\) there exists an \(n\times n\) correlation matrix with characteristic \((c,\sigma)\) and \(w_1<w_{\operatorname{max}}\). 
	\end{theorem}
	\begin{proof}
		By the assumptions  we find that the line through the origin and \((c_0,\sigma_0)\) intersects the unit circle in a point \((c_1,\sigma_1)\) with \(c_1,\sigma_1>0\) and \(\sigma_1>\sqrt{3}c_1\). Then with \(c_1^2+\sigma_1^2=1\) we deduce \(0<c_1<1/2\). By the properties of \(g_n\) (see Lemma~\ref{Lem:gnsnestimate}) we can choose \(n_1\in\N\) such that \(g_n(c_1)<1/2\) for all \(n\geq n_1\). Fix a real number \(\mu\) with
			\begin{equation}0<\mu<\frac{\sqrt{3}}{2}\min\left\{\frac{1}{6}\left(\frac{1}{\sqrt{2}}-\sqrt{g_{n_1}(c_1)}\right),1-\sqrt{\frac{2}{3}}\right\}.\end{equation}
		Choose real numbers \(c_2\) and \(\varepsilon\) with \(c_0<c_2<c_1\) and \(0<\varepsilon<\mu\) such that the open triangle denoted by \(U'\) spanned by \(0\), \((c_2,\sqrt{1-c_2^2})\) and \((c_2,\sqrt{1-c_2^2}-\varepsilon)\) is an open neighborhood around \((c_0,\sigma_0)\). Then choose \(c_3\) with \(c_0<c_3<c_2\) and \(\sqrt{1-c_3^2}-\sqrt{1-c_2^2}<\mu-\varepsilon\). We have that \(U:=\{(c,\sigma)\in U'\mid c<c_3\}\) is an open neighborhood around \((c_0,\sigma_0)\) and it follows from simple geometric observations that given any \(\hat{c},\hat{\sigma}>0\) with \(\hat{c}^2+\hat{\sigma}^2=1\) and \(c_3<\hat{c}<c_2\) we have that \(U\) is contained in the triangle denoted by \(\Delta(\hat{c},\hat{\sigma})\) spanned by the points \(0\), \((\hat{c},\hat{\sigma})\) and \((\hat{c},\hat{\sigma}-\mu)\). By Lemma~\ref{Lem:Thecnical03} we can choose \(n_0\geq n_1\) large enough such that for any \(n\geq n_0\) there exists \(0\leq k\leq n\) with
		\begin{equation}c_3<\frac{n(2k/n-1)^2-1}{n-1}<c_2.\end{equation}
		Then given \(n\geq n_0\) put \(\hat{c}=(n(2k/n-1)^2-1)/(n-1)\) for some \(0\leq k\leq n\) such that \(c_3<\hat{c}<c_2\) holds and set \(\hat{\sigma}=\sqrt{1-\hat{c}^2}\). By Corollary~\ref{Cor:RankOneCorrelation} we find an \(n\times n\) correlation matrix with characteristic \((\hat{c},\hat{\sigma})\) and hence by Proposition~\ref{Lem:PerturbationOfRankOneCorrelation} (note that \(\hat{\sigma}-\mu\geq\hat{\sigma}(1-2/\sqrt{3}\mu)\) holds since \(\hat{\sigma}\geq \sqrt{3}/2\)) and Corollary~\ref{Cor:ConvexityWithId} we find for any given \((c,\sigma)\in\Delta(\hat{c},\hat{\sigma})\) an \(n\times n\) correlation matrix with characteristic \((c,\sigma)\) and \(w_1<w_\text{max}\). Since \(n\geq n_0\) was arbitrary and \(U\) is contained in \(\Delta(\hat{c},\hat{\sigma})\) the claim follows.    
	\end{proof}
From Theorem~\ref{Thm:LocalResultB1} and Theorem~\ref{Thm:LocalResultB2} the conclusion of Theorem~\ref{Thm:AsymptoticCounterExamples} follows immediately.
	\begin{proof}[\textbf{Proof of Theorem~\ref{Thm:AsymptoticCounterExamples}}]
		Let \(K \subset A_2\) be a compact set. We can write \(A_2=B_1\cup B_2\) with \(B_1=\{(c,\sigma)\mid c<\sigma<1-c,\,\,c,\sigma>0\}\) and \(B_2=\{(c,\sigma)\mid \sigma>\sqrt{3}c,\,\,c,\sigma>0,\,\,c^2+\sigma^2<1\}\). Then for any point \(p\in K\) we have  \(p\in B_1\) or \(p\in B_2\). In both cases we find by Theorem~\ref{Thm:LocalResultB1} or Theorem~\ref{Thm:LocalResultB2} an open neighborhood \(U_p\) around \(p\) and a positive integer \(n_p\) such that for any \(n\geq n_p\) and any \((c,\sigma)\in U_p\) there exists an \(n\times n\) correlation matrix with characteristic \((c,\sigma)\) satisfying \(w_1<w_{\text{max}}\). Since \(K\) is compact we find finitely many points \(p_1,\ldots,p_N\) such that \(K\subset U_{p_1}\cup\ldots\cup U_{p_N}\). Then the claim follows for \(n_0:=\max_{1\leq j\leq N}{n_{p_j}}\).
	\end{proof}

\section{Conclusions and Outlook} \label{sec:conclusion}

In the present paper we derived generic features for the spectral structure of correlation matrices in terms of their mean correlation and standard deviation. We showed that some of those properties, earlier observed or conjectured for correlation matrices of large dimension \(n\), are also valid when \(n \geq 2\) is arbitrary.

 Our results provide a quantitative measure to which extent a correlation matrix is approximately given by a single eigenvector depending on its characteristic  $(c,\sigma)$. In particular we discover that not simply a large $c$ but more general large $r_c = \sqrt{c^2+\sigma^2} $ imply distinctly large eigenvalue of the underlying correlation matrix.  Analogously, not only small $\sigma$ but more general small $ \phi_c  = \arccos (c/r_c)$ imply an approximately diagonal eigenvector.

 Furthermore, we explicitly constructed examples of correlation matrices which show that in general  eigenvectors for the leading eigenvalue do not need to be diagonal, or approximately diagonal, even when \(n\) is large. We note that the construction of  correlation matrices with specific spectral properties, is a widely studied but non-trivial task\cite{Chalmers1975,MarsagliaOlkin1984,Numpa2012,TUITMAN2020,Waller2020} itself.
 
In this work we mainly focused on the case \(c>0\), but all bounds in Theorem~\ref{th:l_1_w_max}, except for $w_1$, are also valid for $c\leq 0$. For \(c=0\) we can find correlation matrices with \(w_1<w_\text{max}\) using Corollary~\ref{Cor:ConvexityWithId}.   We note that $c<0 $ is only relevant for small $n$ as  we showed Corollary~\ref{cor:smalles_mean_c}.

In the main results we defined a domain in the \((c, \sigma)\)-plane where \(w_1=w_\text{max}\) generically holds for correlation matrices of any dimension $n$. This domain is a simplifications of the more technical domains from  Theorem~\ref{Thm:w1wmaxDomainGeneral}, which for small $n$ covers a slightly larger area, as shown in Fig.~\ref{Fig:technical_domains}. However, the difference between these domains vanishes when \(n\) becomes large. 

Throughout the paper we considered a correlation matrix as a fixed realisation of a random variable.  Our result contribute to analyses of random or empirical correlation matrices connecting the mean correlation and the standard deviation of the correlation coefficients with the spectral decomposition of the underlying correlation matrix.

Among further studies we expect that our results can be extended  by taking  the skewness and kurtosis of the correlation coefficients into account. Especially, as empirical correlations have been observed to follow a non-stationary and asymmetric distribution (see Fig.~5 in Ref.~\cite{Munnix2012}). Furthermore, we leave the Question~\ref{q:question} on the alignment of the first eigenvector for correlation matrices with  $(c,\sigma) \in A \setminus (A_1 \cup A_2)$ unanswered. 

Finally, we expect that similar results can be obtained for wider class of symmetric positive semi-definite matrices as we mentioned it in Remark~\ref{rem:symmetric_matr}.

\section*{Acknowledgement}

We thank Sebastian Krause and Gerrit Herrmann for fruitful discussions.

\bibliography{refs}

\end{document}